\documentclass[a4paper,UKenglish,cleveref, autoref]{lipics-v2019}
\usepackage{todonotes}

\usepackage{amssymb}

\usepackage{xspace}

\usepackage[shortcuts]{extdash}

\usepackage{tikz}

\usetikzlibrary{shapes,decorations,calc}

\newtheorem{problem}[theorem]{Problem}

\newcommand{\eqdef}{\stackrel{\text{def}}=}

\newcommand{\eqext}[1]{\stackrel{\text{\tiny(#1)}}=}

\newcommand{\dL}{\mathtt{{\scriptstyle L}}}
\newcommand{\dR}{\mathtt{{\scriptstyle R}}}
\newcommand{\dD}{\mathtt{{\scriptstyle D}}}

\newcommand{\fun}[3]{\ensuremath{#1\colon #2 \to #3}}

\newcommand{\mathcalsym}[1]{\ensuremath{\mathcal{#1}}\xspace}

\newcommand{\CA}{\mathcalsym{A}}

\newcommand{\CD}{\mathcalsym{D}}

\newcommand{\CF}{\mathcalsym{F}}
\newcommand{\CG}{\mathcalsym{G}}

\newcommand{\CL}{\mathcalsym{L}}

\newcommand{\CO}{\mathcalsym{O}}
\newcommand{\CP}{\mathcalsym{P}}
\newcommand{\CQ}{\mathcalsym{Q}}
\newcommand{\CR}{\mathcalsym{R}}
\newcommand{\CS}{\mathcalsym{S}}

\newcommand{\stany}{\mathrm{states}}
\newcommand{\Gw}{\omega}

\newcommand{\domain}[1]{\ensuremath{\mathbb{#1}}\xspace}

\newcommand{\N}{\domain{N}}

\newcommand{\R}{\domain{R}}

\newcommand{\limit}{\mathrm{limit}\xspace}

\newcommand{\defPlayer}[1]{\ensuremath{{\boldsymbol{#1}}}\xspace}

\newcommand{\eve}{\defPlayer{\exists}}
\newcommand{\adam}{\defPlayer{\forall}}

\newcommand{\dom}{\mathrm{dom}\xspace}

\newcommand{\powerset}{\mathsf{P}}

\newcommand{\init}{\mathrm{I}\xspace}
\newcommand{\lang}{\mathrm{L}\xspace}

\newcommand{\trees}{\mathrm{Tr}\xspace}
\newcommand{\runs}{\mathrm{Runs}\xspace}

\newcommand{\restr}{{\upharpoonright}}

\newcommand{\Fdis}{\CD}
\newcommand{\BC}{\mathrm{BC}\xspace}

\newcommand{\Fmu}{\vec{\mu_0}}

\usepackage{tikz-cd}

\newcommand{\evalInt}[2]{%
	\pgfmathparse{int(#2)}%
	{\global\edef#1{\pgfmathresult}}%
}

\newcommand{\evalFloat}[2]{%
	\pgfmathparse{#2}%
	{\global\edef#1{\pgfmathresult}}%
}

\title{Computing measures of weak-MSO definable sets of trees} %TODO Please add

\titlerunning{Computing measures of weak-MSO definable sets of trees}%optional, please use if title is longer than one line

\author{Damian Niwiński}{Institute of Informatics, University of Warsaw}{niwinski@mimuw.edu.pl}{}{}
\author{Marcin Przybyłko}{Fachbereich Informatik, University of Bremen}{przybyl@uni-bremen.de}{}{}
\author{Michał Skrzypczak}{Institute of Informatics, University of Warsaw}{mskrzypczak@mimuw.edu.pl}{0000-0002-9647-4993}{Polish National Science Centre 	grant 2016/22/E/ST6/00041.}

\authorrunning{D. Niwiński, M. Przybyłko, and M. Skrzypczak}%TODO mandatory. First: Use abbreviated first/middle names. Second (only in severe cases): Use first author plus 'et al.'

\Copyright{Damian Niwiński, Marcin Przybyłko, and Michał Skrzypczak}%TODO mandatory, please use full first names. LIPIcs license is "CC-BY";  http://creativecommons.org/licenses/by/3.0/


%TODO mandatory; please add comma-separated list of keywords

\category{}%optional, e.g. invited paper

\relatedversion{}%optional, e.g. full version hosted on arXiv, HAL, or other respository/website
\supplement{}%optional, e.g. related research data, source code, ... hosted on a repository like zenodo, figshare, GitHub, ...

\nolinenumbers %uncomment to disable line numbering

\hideLIPIcs  %uncomment to remove references to LIPIcs series (logo, DOI, ...), e.g. when preparing a pre-final version to be uploaded to arXiv or another public repository

\EventEditors{John Q. Open and Joan R. Access}
\EventNoEds{2}
\EventLongTitle{42nd Conference on Very Important Topics (CVIT 2016)}
\EventShortTitle{CVIT 2016}
\EventAcronym{CVIT}
\EventYear{2016}
\EventDate{December 24--27, 2016}
\EventLocation{Little Whinging, United Kingdom}
\EventLogo{}
\SeriesVolume{42}
\ArticleNo{23}
\usepackage{cite}
\begin{document}

\maketitle

\begin{abstract}
%
%This work addresses the problem of computing probability of sets of
%infinite trees with respect to the standard coin-flipping measure.
%An algorithm is provided  to compute the probability of tree
%languages recognizable by weak alternating automata, or equivalently
%definable in weak monadic second-order logic.  This class of
%languages, although smaller than all regular tree languages, is still
%fairly comprehensive. In particular, the solution works for tree
%languages definable in first\=/order logic as well as those definable
%in CTL.  Thus the new algorithm enhances the tollbox of quantitative
%model checking.
%
%\noindent
This work addresses the problem of computing measures of recognisable sets of infinite trees. An~algorithm is provided  to compute the probability measure of a~tree language recognisable by a~weak alternating automaton, or equivalently definable in weak monadic second\=/order logic. The measure is the uniform coin\=/flipping measure or more generally it is generated by a~branching stochastic process. The class of tree languages in consideration, although smaller than all regular tree languages, comprises in particular the languages definable in the alternation\=/free $\mu$\=/calculus or in temporal logic CTL. Thus, the new algorithm may enhance the toolbox of  probabilistic model checking.
\end{abstract}
\newpage

\section{Introduction}
\label{sec:intro}

%\MS{I like this new introduction very much. It is to the point and makes the result seem strong :)}

The non\=/emptiness problem asks if~an~automaton accepts at~least one object. From a~logical perspective, it~is an~instance of~the consistency question: does a~given specification have a~model? Sometimes it is also relevant to ask a~quantitative version of~this question: whether a~\emph{non\=/negligible} set of models satisfy the specification. When taken to the realm of probability theory, this boils down to estimating the probability that a~random object is~accepted by~a~given automaton. In~this paper, models under consideration are infinite binary trees labelled by~a~finite alphabet. Our main problem of~interest is~the following.
\begin{problem}
\label{pro:comp-mso}
Given a~regular tree language $L$, compute the probability that a~randomly generated tree belongs to~$L$.
\end{problem}

\noindent
In~other words, we~ask for the probability measure of~$L$.
Here, the tree language~$L$ might be~given by~a~formula of~monadic second\=/order logic, but for complexity reasons it~is~more suitable to~present it~by
a~tree automaton or~by~a~formula of~modal $\mu$\=/calculus, see e.g.~\cite{gradel_automata,demri-temporal-book}.
By default, the considered measure is~the uniform \emph{coin\=/flipping} measure, where each letter is~chosen independently at~random; but also more specific measures are of~interest. If~the computed probability is~rational    then it~can 
be
represented explicitly, but the measure can~be irrational, see e.g.~\cite{michalewski_comp_measure}, and may require more complex representation.  One of~the possible choices, exploited in~this paper, is~a~formula over the field of~reals~$\R$.

Chen et~al.~\cite{chen_model_checking} addressed Problem~\ref{pro:comp-mso} in~the case where the tree language~$L$ is~recognised by~a~deterministic top\=/down automaton and the measure is~induced by~a~stochastic branching process, which then makes also a~part of~the input data. Their algorithm compares the probability with any~given rational number in~polynomial space and with $0$ or~$1$ in~polynomial time. The limitations of~this result come from the deterministic nature of~the considered automata: deterministic top\=/down tree automata are known to~have limited expressive power within all regular tree languages.

Michalewski and Mio~\cite{michalewski_comp_measure} stated Problem~\ref{pro:comp-mso} explicitly and solved it~for languages~$L$ given by~so\=/called \emph{game automata} and the coin\=/flipping measure.
This class of~automata subsumes deterministic ones and captures some important examples including the game languages, cf.~\cite{murlak_game_auto}, but even here the strength of~non\=/determinism is~limited; in particular, the class is not closed under finite union. The algorithm from~\cite{michalewski_comp_measure} reduces the problem to~computing the value of~a~Markov branching play, and uses Tarski's decision procedure for the theory of~reals. These authors also discover that the measure of~a~regular tree language can be~irrational, which stays in~contrast with the case of~$\omega$\=/regular languages, i.e.~regular languages of~infinite words, where the coin\=/flipping measure is~always rational, cf.~\cite{stochastic-parity-soda-04}.
%\smallskip

%\MS{Ja bym byl za stroceniem kolejnego akapitu by zasygnalizowac wyniki~\cite{przybylko_first_order} a potem powiedziec onkretniej tylko o non-det safety. Ale nie wiem co myslicie?}

%\DN{Sprobuje, w pierwszej wersji tak bylo.}

%\todo{referencja do Irańczyka}

Another step towards a~solution to Problem~\ref{pro:comp-mso} was 
%given 
made
by the second author of~the present article, who proposed an~algorithm to~compute the coin\=/flipping measure of~tree languages definable in fragments of first\=/order logic~\cite{przybylko_first_order}. 
% The algorithm settles 
%in triple exponential time 
%the cases when the language~$ L$ is~given by~a~formula involving monadic predicates and  child relation but not descendant relation, or~by~a~%and in exponential space the case of 
%Boolean combinations of~conjunctive queries (with descendant relation) ~\cite{przybylko_first_order}.  The respective classes of~tree languages,  although relatively limited, are incomparable with those definable by~game automata.
%The author also demonstrates an~example of~a~tree language defined by~a~first\=/order formula (with descendant) having irrational probability.
This work is~subsumed in~a~report~\cite{przybylko_fo_arxiv} (accepted for publication in~a~journal) co\=/authored with~the third author, where a~new class of languages~$L$ is also resolved:
tree languages recognised by safety automata, i.e.~non\=/deterministic automata with a~trivial accepting condition.

An~analogue of~Problem~\ref{pro:comp-mso} can be stated for $\omega$\=/regular languages.  As~noted by~\cite{chen_model_checking}, the problem then reduces to~a~well\=/known question in~verification solved by~Courcoubetis and Yannakakis~\cite{courcou-yanna-acm} already in~the 1990s,  namely whether a~run of a~finite\=/state Markov chain satisfies an~$\omega$\=/regular property. The algorithm runs in~single\=/exponential time w.r.t.~the automaton (and linear w.r.t.~the Markov chain). A~related question was also studied by~Staiger~\cite{staiger_computable_98}, who gave an~algorithm to~compute Hausdorff dimension and Hausdorff measure of a~given $\omega$\=/regular language.

In general, Problem~\ref{pro:comp-mso} remains unsolved. At~first sight, one may even wonder if it is well\=/stated, as~regular tree languages need not in~general be Borel, cf.~\cite{niwinski_gap}.
% (in contrast to the MSO\=/definable sets of infinite words, which are on the 3rd level of the Borel hierarchy~\cite{landweber_69}); it 
However, due to~\cite{mio_game_semantics,michalewski_measure_final}, we~know that regular languages of~trees are always universally measurable.

In~the present paper, we~solve Problem~\ref{pro:comp-mso} in~the case where the language~$L$ is~recognised by~a~weak alternating automaton or,~equivalently, defined by~a~formula of~weak monadic second\=/order logic, cf.~\cite{muller_alternating_complexity}.  
The class of~tree languages in~consideration is~incomparable with the one considered by~Michalewski and Mio~\cite{michalewski_comp_measure}, but subsumes those considered in~\cite{przybylko_first_order,przybylko_fo_arxiv}. Yet another presentation of~this class can be~given in~terms of~alternation\=/free fragment of modal $\mu$\=/calculus, see~\cite{arnold_fixed_weak} for details. This fragment is known to be useful in~verification and model checking, 
%For instance, the present result applies to the logic~CTL as~it~can be~encoded in~modal $\mu$\=/calculus.
in particular, temporal logic CTL embeds into this fragment.  

We~consider the coin\=/flipping measure as~our primary case, but we~also show how to~extend our approach to~measures generated by~stochastic branching processes, as~in~\cite{chen_model_checking}. The computed probability is~presented by~a~first\=/order formula in~prenex normal form over the field of~reals. The provided formula is~exponential in~the size of~the automaton and polynomial in~the size of~the branching process. Moreover, the quantifier alternation of~the computed formula is constant (equal $4$). Combined with the known decision procedures for the theory of~reals, this gives the following.
\begin{theorem}
\label{thm:main-branching}
There 
%exists 
is
an~algorithm that inputs a~weak alternating parity automaton $\CA$, a~branching process $\CP$, and a~rational number $q$ encoded in binary; and decides if~the measure generated by~$\CP$ of~the language recognised by~$\CA$ is~equal, smaller, or greater than $q$. The algorithm works in~time polynomial in $q$, doubly exponential in~$\CA$, and singly exponential in~$\CP$.
\end{theorem}

\begin{comment}
Similarly to the approach taken in~\cite{przybylko_fo_arxiv}, we~reduce the problem to~computation of~an~appropriate probability distribution over the powerset of~the automaton's states. To~do~so, we~use the structure $\langle \Fdis{\powerset(Q)}, \preceq \rangle$ of~all such distributions with a~suitable ordering. The key idea is~to~approximate the target  language~$L$ by~two families of~series of~tree languages, representing safety and reachability properties,  respectively. The measures of~these approximations are computed by~mutual induction
with a~repeated use of~limit constructions. Our ability to~perform these limit constructions is~based on a~synergy between the order and topological properties of~$\Fdis{\powerset(Q)}$.
%
Similarly to the approach taken in~\cite{przybylko_fo_arxiv}, we~reduce the problem to~computation of~an~appropriate probability distribution over the powerset of~the automaton's states. To~do~so, we 
consider 
 the set of~all such distributions $\Fdis{\powerset(Q)} $ with  a suitable ordering $\preceq  $.
The transformation of ordered sets $\langle \powerset(Q) , \subseteq \rangle $ $\mapsto $
$\langle \Fdis{\powerset(Q)}, \preceq \rangle$ 
can be in fact  seen  as an~instance of 
a~functor,
%related to  a more general
but we do not use the category-theoretic language in this paper (for related concepts see~\cite{Plotkin-proba-89}).  
\end{comment}

Similarly to the approach taken in~\cite{przybylko_fo_arxiv}, we~reduce the problem to~computation of~an~appropriate probability distribution over the powerset of~the automaton's states. To~do~so, we 
consider 
 the set of~all such distributions $\Fdis{\powerset(Q)} $ with  a suitable ordering $\preceq  $.
The structure is in fact a finitary case of a probabilistic powerdomain introduced by 
Saheb-Djahromi~\cite{saheb-powerdomain} (see also~\cite{Plotkin-proba-89}),
but we do not exploit  category-theoretic concepts in this paper.
The key step is an approximation of the target  language~$L$ by~two families of~tree languages  representing safety and reachability properties,  respectively. Then we can apply  fixed-point constructions thanks to a kind of synergy between the order and topological properties of~$\Fdis{\powerset(Q)}$.

%\section{Basic concepts}
\section{Trees, topology, and measure}
\label{sec:basic}

The set of natural numbers $\{0, 1, 2, \ldots \}$  is denoted  by $\N$,
or by~$\omega$ whenever we treat it~as an~ordinal.  
A~finite non\=/empty set $A$ is called an~\emph{alphabet}. By $\powerset(X)$ we denote the family of all subsets of a~set $X$.
The set of \emph{finite words} over an~alphabet~$A$ (including the \emph{empty
word} $\varepsilon$) is denoted by $ A^*$,  and the set of $\omega
$\=/words by~$A^{\omega }$. The \emph{length} of a~finite word $w\in A^*$ is denoted by~$|w|$.
A~\emph{full infinite binary tree} over an~alphabet~$A $ (or simply a~\emph{tree} if
confusion does not arise) is a~mapping~$\fun{t}{\{ \dL , \dR \}^*}{A}$.
The set of all such trees, denoted by  $\trees_A $, 
can be~equipped with a~topology induced by~a~metric 
\begin{equation*} \label{distance}
d( t_1, t_2 ) = \left\{\begin{array}{ll}
0
&\mbox{ if $t_1 = t_2 $}\\
2^{-n} \mbox{ with } n = \mbox{min} \{ |w| \mid t_1 (w) \neq t_2 (w) \}
& \mbox{ otherwise. }
\end{array}\right.
\end{equation*}
It is well\=/known that this topology coincides with the product
topology on $A^{\omega} $, where $A$ is~a~discrete topological space. 
The topology can be generated by~a~basis consisting of all the sets $U_f$, where $\fun{f}{\dom(f)}{A}$ is a function with a finite domain $ \dom(f)\subset \{ \dL , \dR \}^*$,
and $U_f$ consists of all trees $ t$ that coincide with $ f$  on  $
\dom(f)$.  If $ A$ has at
least $2$ elements then this topology is homeomorphic to the Cantor
discontinuum $\{0,1\}^{\omega } $ (see, e.g.~\cite{perrin_pin_words}).

The set of trees can be further equipped with a~probabilistic measure $\mu_0
$, which is~the standard Lebesgue measure on the product space 
defined on the basis by $\mu_0 \left(  U_f \right) =  \big|A\big|^{-| \dom(f)|}$.

We note a  useful property   of this measure, which intuitively
amounts to saying that
events happening  in incomparable nodes are
independent.  For $t\in\trees_A$ and $ v \in \{ \dL , \dR \}^*$, 
the subtree of~$ t$ induced by $ v$  is a tree $t\restr_v
\in\trees_A$ defined by $t\restr_v (w) = t
(vw)$,   for $w \in  \{ \dL , \dR \}^*$.
\begin{remark}\label{rem:measure}
If $ v_1, \ldots , v_k \in \{ \dL , \dR
\}^* $  are pairwise incomparable  nodes (i.e., none is a prefix of another)  
and
$ V_1, \ldots , V_k \subseteq \trees_A $  are  Borel sets then
\begin{eqnarray}\label{equality:measure}
\mu_0 \left(   
\{  t\in\trees_A  \mid   t\restr_{v_i} \in V_i \mbox{ for  } 
i = 1, \ldots, k
\} \right)& = &
\mu_0 (V_1) \cdot \ldots \cdot \mu_0 (V_k).
\end{eqnarray}
\end{remark}
\begin{comment}
%\todo{Review 1: ``124--133 could be omitted''}
Note that, for $ k=1$, the above equality means that
$\mu_0 \left(   
\{  t\in\trees_A  \mid   t\restr_{v_1} \in V_1 
\} \right) = \mu_0 (V_1)$, which  can be easily verified for the basis sets
and then shown  by induction for all Borel sets thanks to the continuity of
the measure.   Using this case,  the right-hand side of 
Equality~\eqref{equality:measure}  amounts to the product of the measures
$\mu_0 \left(   
\{  t\in\trees_A  \mid   t\restr_{v_i} \in V_i
\} \right) $, for $i = 1,\ldots ,k $.  Then Equality~\eqref{equality:measure} can 
 be shown,  again by induction on Borel
sets, where the basic  case relies on the following observation.
If $\fun{f_1}{\dom(f_1)}{A}$ and
$\fun{f_2}{\dom(f_2)}{A}$ have disjoint (finite) domains, then
\begin{eqnarray*}
\mu_0\left(U_{f_1}\cap U_{f_2}\right)=
\mu_0 \left( U_{f_1 \sqcup f_2} \right) =
\big|A\big|^{- (  |\dom (f_1)| + |\dom (f_2)|)} = 
\mu_0 \left( U_{f_1} \right) \cdot \mu_0 \left( U_{f_2} \right) .
\end{eqnarray*}

%Moreover, we require $\mu_0$ to treat particular positions of the
%tree independently: if $\fun{f_1}{\dom(f_1)}{A}$ and
%$\fun{f_2}{\dom(f_2)}{A}$ have disjoint domains, then
%$\mu_0\left(U_{f_1}\cap U_{f_2}\right)=\mu_0\big(U_{f_1}\big)\cdot
%\mu_0\big(U_{f_2}\big)$.

\noindent
\end{comment}
\noindent
We refer to e.g.~\cite{michalewski_measure_final} for more 
detailed considerations of~measures on sets of~infinite trees.

%\paragraph{Partial orders.}
%By $X$, $Y$, $Z$ we will denote finite partially orders sets, where the order relation is ${\leq}$. A function $\fun{f}{X}{Y}$ is \emph{monotone} if $\forall x\leq y\in X.\ f(x)\leq f(y)$. A~subset $U\subseteq X$ is called \emph{upward\=/closed} if $\forall x\leq y\in X.\ x\in U\rightarrow y\in U$. For a~pair of functions $\fun{f,g}{X}{Y}$ we say that $f\leq g$ if $\forall x\in X.\ f(x)\leq g(x)$.

\section{Tree automata and games}

An~{\em alternating parity automaton\/} over infinite trees can be~presented
as~a~tuple $\CA=\langle A, Q, q_\init, \delta, \Omega\rangle$, where
$A$ is~a~finite alphabet; $Q$  a~finite set of~\emph{states}; $q_\init\in Q$
an~\emph{initial state};  $\fun{\delta}{Q\times
  A}{\BC^{+}\big(\{\dL,\dR\}\times Q\big)}$  a~\emph{transition function}
that assigns to~a~pair~$(q,a)\in Q\times A$ a~finite positive Boolean
combination of~pairs~$(d,q')\in\{\dL,\dR\}\times Q$; and finally
$\fun{\Omega}{Q}{\N}$ is a~\emph{priority mapping}.% For the sake of succinctness, we allow expressions $\bot$ and $\top$ in the Boolean combinations within $\delta$, but it does not matter much.

In~this paper, we~assume that automata are {\em weak\/},
i.e.~the priorities~$\Omega(q)$ are non\=/increasing along
transitions. More precisely, if~$(d,q')$ is~an~atom that appears in~the
formula~$\delta(q,a)$ then $\Omega(q)\geq\Omega(q')$. Given
$n\in\N$, we~denote by~$Q_{<n}$ and $Q_{\geq n}$ the subsets of~$Q$
consisting of~those states whose priority is~respectively strictly smaller or~greater than~$n$.

The semantics of~an~automaton can be~given in~a~terms of~a~game
%$\CG(t,p)$ 
played by~two players \eve and \adam over a~tree~$t$ in~$\trees_A $ from a~state~$p \in Q$.
Let $\Gamma$ be~the set of~all sub\=/formulae of~the formulae in~$\delta(q,a)$,
for all~$(q,a)\in Q\times A$. 
%Then  
The set of~\emph{positions}
%$\CG(t,p)$ 
of~the game is~the set~$\big(Q\sqcup \Gamma\big)\times \{\dL,\dR\}^\ast$ and the
\emph{initial position} is~$\big(p,\varepsilon\big)$. The positions of~the
form $\big(q,v\big)$, $\big(\phi_1{\lor}\phi_2, v\big)$, and
$\big((d,q), v\big)$ are controlled by~\eve, while the positions of~the
form~$\big(\phi_1{\land}\phi_2, v\big)$ are controlled by~\adam.
The edges connect the following types of~positions:

\begin{itemize}
\item $\big(q,v\big)$ and $\big(\delta(q, t(v)), v\big)$,
\item $\big(\phi_1{\lor}\phi_2, v\big)$ and $\big(\phi_i, v\big)$ for $i=1,2$,
\item $\big(\phi_1{\land}\phi_2, v\big)$ and $\big(\phi_i, v\big)$ for $i=1,2$,
\item $\big((d,q), v\big)$ and $\big(q, v\cdot d\big)$.
\end{itemize}

\noindent
We~assume that every formula in~the image~$ \delta(Q\times A)$ is~non\=/trivial and, thus, every position is~a~source of~some edge.
\smallskip

The directed graph described above forms  the  {\em arena\/} of~our game
that we~denote by~$\CG(t,p) $.  A~{\em play\/} in~the arena is~any
infinite path starting from the initial position~$\big(p,\varepsilon\big)$.
We~call the positions of~the form~$(q,v) $  {\em state positions\/}.
Given a~play~$ \pi$,
the \emph{states} of~the play denoted $\stany \, (\pi )$  is~the
sequence of~states $(q_0, q_1, \ldots)\in Q^\Gw$ such that the
successive state positions visited during $\pi$ are $\big(q_i, v_i\big)$,  for
$i=0,1,\ldots$,  and some 
%$(v_i)_{i\in\N}$. 
$(v_i)_{i \in  \omega }$. 
\medskip

To complete the definition of~the game, we~specify a~winning
criterion for~\eve.  The default  is~the \emph{parity condition}, but we~will also  consider other
criteria. 
Let 
\begin{eqnarray*}
\runs\eqdef\{\rho\in Q^\Gw\mid \forall i\in\omega.\
\Omega(\rho(i))\geq \Omega(\rho(i{+}1))\big\}
\end{eqnarray*}
be~the set that contains 
all sequences of states that induce non\=/increasing sequences of priorities.
Notice that since $\CA$ is weak, only such sequences may arise in the
game.
In~general, a~\emph{winning condition} is~any set~$W\subseteq \runs$. That is,  a~play~$\pi$ is~\emph{winning} for \eve  with
respect to~$W$ if, and only if, 
% the sequence of states $(q_0,q_1, \ldots)\in Q^\Gw$ of $\pi$ satisfies
$\stany \, (\pi ) \in W$. The game with a~winning set~$W$ is~denoted by~$\CG(t,p,W)$.

\smallskip
The \emph{parity condition} $W_P \subseteq \runs $ for a~weak
automaton  amounts to saying that
$(q_0,q_1,
\ldots)\in W_P$ if~$\lim_{i\to\infty} \Omega(q_i)\equiv 0 \mod 2$,
i.e.~the limit priority of~states visited in~a~play is~even. 
Let
$\lang(\CA,p)$ be~the set
 of~trees such that \eve has a~winning strategy in~$\CG(t,p, W_P)$. 
Then, the language  of~an~automaton~$\CA$ is~the set
$\lang(\CA) \eqdef \lang(\CA,q_\init)$, where $q_\init$ is~the initial state of~$\CA$.
\medskip

As~mentioned above, we~will consider  games with various winning
criteria. The following simple observation is~useful.

\begin{remark}
\label{rem:win-monotone}
If $W\subseteq W'\subseteq \runs$ then the following implication holds: if \eve wins $\CG(t,p,W)$ then \eve wins also $\CG(t,p,W')$.
\end{remark}

Since the winning criteria in~consideration 
will always be~$\omega $\=/regular languages of~infinite words, we~implicitly rely on~the following classical
fact~(cf.~\cite{gradel_automata}).

\begin{proposition}
\label{pro:determinacy}
Games on graphs with $ \omega$\=/regular winning conditions are finite
memory determined.
\end{proposition}

\noindent 
We~will also use the following fact, cf.\ e.g.\ \cite{muller_alternating_complexity,skurczynski_borel_infinite}.

\begin{proposition} 
\label{pro:measurability}
For a~weak alternating parity automaton~$ \CA$, all tree languages
$\lang(\CA,p)$ are Borel and, consequently, measurable with respect to~the uniform measure~$ \mu_0$ (and also any other Borel measure on $\trees_A$).
\end{proposition}
Note that measurability holds for
non\=/weak automata as well~\cite{michalewski_measure_final}.

\section{Approximations}
\label{sec:recognition}

For the sake of~this section we~fix a~weak alternating parity automaton~$\CA$.
Our aim is~to~provide some useful  approximations
for the tree languages~$\lang(\CA, p)$.
The approximations are simply some~families of~tree languages indexed by~states~$p\in Q$.
Those families, called \emph{$Q$\=/indexed families}, or \emph{$Q$\=/families} for short,
are~represented by~functions~$\fun{\CL}{Q}{\powerset(\trees_A)}$.
% that assigns to each state $q\in Q$ a~set of trees $\CL(q)\subseteq\trees_A$.  
By the construction, we will guarantee that the tree languages~$\CL(q)$ 
will themselves be recognisable by~some weak alternating automata.
%By an obvious correspondence between the function spaces $\left(  2^X\right)^Y $ and $\left( 2^Y\right)^X $, 
Each $ Q$\=/family naturally possesses
a~dual representation by~a~mapping~$\trees_A \to \powerset(Q) $ that we~denote by~the
same letter (but with different brackets)
\[\CL[t]\eqdef \{q\in Q\mid t\in \CL(q)\}\in\powerset(Q).\]
%
%Take $p\in Q$ and $n, i\in \N$. [ p does not occur below ]
If~$\rho\in \runs\subseteq
Q^\Gw$ is~an~infinite sequence of~states then $\lim_{i\to\infty} \Omega(\rho(i))$ (denoted by $\limit(\rho)$) exists, because
by the definition of~$\runs$ the priorities are non\=/increasing and bounded.
Recall that $W_P\subseteq \runs$ is~the set of~runs satisfying the
parity condition, i.e.~$W_P=\{\rho\in\runs\mid\limit(\rho)\equiv 0\mod
2\}$. 
For $i,n \in \N $, consider the following subsets of~$\runs$:
\begin{align*}
S^n_i&\eqdef W_P\cup\big\{\rho\in \runs\mid \Omega(\rho(i))\geq n\big\},\\
S^n_\infty&\eqdef W_P\cup\big\{\rho\in \runs\mid \limit(\rho)\geq n\big\},\\
R^n_i&\eqdef W_P\cap\big\{\rho\in \runs\mid \Omega(\rho(i))< n\big\},\\
R^n_\infty&\eqdef W_P\cap\big\{\rho\in \runs\mid \limit(\rho)< n\big\}.
\end{align*}
Connotatively, the name of~the sets~$S^n_i$ comes from the condition of~\emph{safety}, while the sets~$R^n_i$ are named after \emph{reachability}. More precisely, $S^n_i$ is an~over\=/approximation of $W_P$, that makes~\eve win also if she manages to reach a~priority $\geq n$ in the $i$th visited node of a~given tree. Analogously, $R^n_i$ is an~under\=/approximation of $W_P$ that makes~\adam win in the above case.

Based on~the above definitions, we~define the respective
$Q$\=/families. 
%Take $n\in \N$, $p\in Q$, and $i\in\N$. 
For~$p\in Q $,  let $\CS^n_i(p)$, $\CS^n_\infty(p)$, $\CR^n_i(p)$, and
$\CR^n_\infty(p)$ be~the sets of~trees such that \eve has a~winning
strategy 
in~the game~$\CG(t,p,W)$, where $W$ is~respectively $S^n_i$,
$S^n_\infty$, $R^n_i$, and~$R^n_\infty$. Figure~\ref{fig:construction} 
below depicts the way these $Q$\=/families are used in the general construction.

%\begin{remark}
%All the above properties are regular, so each of the above
%definitions gives a~$Q$\=/family of languages.
It~is~easy to~see that all the tree languages above can be~recognised by~weak parity alternating automata.
%\end{remark}

\begin{lemma}
\label{lem:S-R-monotone}
For every $n\in \N$ and $i\in\N$,  we have
\[ S^n_i\supseteq S^n_{i+1} \supseteq S^n_\infty\text{ and } R^n_i\subseteq R^n_{i+1} \subseteq R^n_\infty.\]
Analogously, for every $p\in Q$,
\[\CS^n_i(p) \supseteq \CS^n_{i+1}(p)\supseteq \CS^n_\infty(p)\text{ and }\CR^n_i(p) \subseteq \CR^n_{i+1}(p)\subseteq \CR^n_\infty(p).\]
\end{lemma}

\begin{proof}
The first property follows directly from the definition of~$\runs$, which guarantees that $\Omega(\rho(0))\geq \Omega(\rho(1))\geq \ldots \geq \limit(\rho)$. Then, the second property follows from Remark~\ref{rem:win-monotone}.
\end{proof}
It~is~straightforward to~see that $S^n_{\infty}= \bigcap_{i\in\N} S^n_i
$ and $ R^n_{\infty}= \bigcup_{i\in\N} R^n_i$.
However, it is not clear that these equalities imply the desired properties for the respective sets
of~trees.  Lemma~\ref{lem:limit-steps-S-R} below implies that
it~is~the case. The proof relies on~combinatorics of~binary trees, namely on~Kőnig's Lemma.

\begin{lemma}
\label{lem:konig-for-strat}
Take $n\in\N$ and $ p \in Q$.
Let 
$B^n_\infty=\{\rho\in\runs\mid \limit(\rho) < n\}$
and,
for $i\in\N$, let
$B^n_i=\{\rho\in\runs\mid \Omega(\rho(i)) < n\}$. If $\sigma$ is
a~winning strategy of \eve in $\CG(t,p,B^n_\infty)$ then there exists
a~number $J\in\N$, such that $\sigma$ is actually winning in
$\CG(t,p,B^n_{J})$. An analogous property holds if $ \sigma $ is
a~winning strategy for \adam.
\end{lemma}

\begin{proof}
Let $\sigma$ be~a~winning strategy of~\eve in~$\CG(t,p,B^n_\infty)$  (the case of~\adam is~completely analogous).
Let $T\subseteq\big(Q\times\{\dL,\dR\}\big)^\ast$ be~the set of~sequences
$(q_i,d_i)_{i\leq \ell}$, with $ q_0 = p$,  such that there exists a~play consistent with
$\sigma$ that visits all the positions $(q_i,d_0\cdots d_{i-1})$ for
$i=0,1,\ldots,\ell$,  and additionally $\Omega(q_\ell)\geq n$. Observe that
$T$ is~prefix\=/closed. Thus, we~can treat $T$ as a~tree. Moreover, as~$Q\times\{\dL,\dR\}$ is~finite, $T$ is~finitely branching. If $T$ is~finite then there exists~$J$ such that all the sequences in~$T$ have length at~most~$J$. In~that case
$\sigma$ is~winning in~$\CG(t,p,B^n_{J})$, and we~are done.

%Assume contrarily, 
For the sake of~contradiction, suppose
that $T$ is~infinite. Apply K\"onig's Lemma to~obtain an~infinite path
$(q_i,d_i)_{i\in\omega }$ in~$T$. By~the definition of~$T$, it~implies that there exists an~infinite play consistent with~$\sigma$ such that $(q_i)_{i\in\omega}$ is~the sequence of~states visited during the play. But this is~a~contradiction, because $\limit\big((q_i)_{i\in\omega}\big)\geq n$ by~the definition of~$T$ and, therefore, the considered play is~losing for~\eve.
\end{proof}

\begin{lemma}
\label{lem:limit-steps-S-R}
Using the above notions, for every state $p\in Q$,  we~have
\[\CS^n_{\infty}(p) = \bigcap_{i\in\N}\CS^n_i(p)\quad\text{and}\quad \CR^n_{\infty}(p) = \bigcup_{i\in\N}\CR^n_i (p).\]
\end{lemma}

\begin{proof}
Consider the first claim and take a~tree $t\in\trees_A$ such that for
every $i\in\N$ we have $t\in\CS^n_i(p)$. We~need to~prove that
$t\in\CS^n_\infty(p)$. Assume contrarily, that
$t\notin\CS^n_\infty(p)$. By~determinacy, see
Proposition~\ref{pro:determinacy}, it~means that there exists
a~strategy~$\sigma'$ of~\adam such that for every play~$\pi$
consistent with~$\sigma'$,
% that visits states $(q_0,\ldots)$ 
we~have 
%$ \limit\big((q_j)_{j\in\N}\big)<n$
$\limit ( \stany \, (\pi ) ) < n $
and $\limit (\stany \, (\pi )) $ is~odd. 
Hence, in~particular,  $ \sigma'$ is~winning for \adam 
in~$\CG(t,p,B^n_\infty)$.   Therefore, by~Lemma~\ref{lem:konig-for-strat},
we~know that there exists a~number~$J\in \N$ such that, for every $
\pi $ consistent with $\sigma' $  with
$\stany \, (\pi ) = (q_0,q_1,\ldots ) $, 
we~have $\Omega(q_{J}) < n$. Therefore, the strategy $\sigma'$ witnesses that $t\notin \CS^n_{J}(p)$, a~contradiction.
\smallskip

We 
%will 
now prove the second claim. Take a~tree~$t\in\CR^n_\infty(p)$. We~need
to~prove that $t\in \CR^n_i(p)$ for some $i\in \N$. Let $\sigma$ be~a~strategy
of~\eve witnessing that $t\in\CR^n_\infty(p)$. Again,
Lemma~\ref{lem:konig-for-strat} guarantees that there exists a~number
$J\in\N$ such that if~$\pi$ is~a~play consistent with $\sigma$ 
with $\stany \, (\pi ) = (q_0, q_1, \ldots) $
then $\Omega(q_J)< n$. Thus, $t\in \CR^n_J(p)$.
\end{proof}

The following lemma provides another characterisation of the above $Q$\=/families.

\begin{lemma}
\label{lem:S-R-limit-steps}
For each $p\in Q$,  we~have $\CS^0_i(p)=\trees_A$ and
$\CR^0_i(p)=\emptyset$. Take $n>0$. If $\Omega(p)\geq n$ then
$\CS^n_0(p)=\trees_A$ and $\CR^n_0(p)=\emptyset$. 
%On the other hand, if 
If~$\Omega(p)<n$ then
\begin{equation*}
\begin{split}
\lang(\CA,p)&=\CS^n_0(p),\\
\lang(\CA,p)&=\CS^{n-1}_\infty(p)&\text{for odd $n$,}
\end{split}
\qquad\quad\ 
\begin{split}
\lang(\CA,p)&=\CR^n_0(p),\\
\lang(\CA,p)&=\CR^{n-1}_\infty(p)&\text{for even $n$.}
\end{split}
\end{equation*}
\end{lemma}

\begin{proof}
The cases of~$n=0$ are trivial. The first two claims in~the case $\Omega(p)\geq n$ follow directly from the definitions. Take $p$ such that $\Omega(p)<n$. Notice that in~that case the sequence of~states $\rho$ in~a~play in~$\CG(t,p)$ satisfies
\begin{equation*}
\begin{split}
\rho\in W_P&\Longleftrightarrow\rho\in S^n_0,\\
\rho\in W_P&\Longleftrightarrow\rho\in S^{n-1}_\infty&\text{for odd $n$,}
\end{split}
\qquad
\begin{split}
\rho\in W_P&\Longleftrightarrow\rho\in R^n_0,\\
\rho\in W_P&\Longleftrightarrow\rho\in R^{n-1}_\infty&\text{for even $n$.}
\end{split}
\end{equation*}
where the first two equivalences follow from the fact that
$\Omega(\rho(0))=\Omega(p)<n$. The last two equivalences can
be~derived from the fact that $\limit(\rho)\leq \Omega(p)< n$. 
First, we~have
%if $n$ is odd then $n{-}1$ is even and therefore 
$\limit(\rho)\geq n{-}1\Leftrightarrow \limit(\rho)=n{-}1$. Thus,
if~$n$ is~odd and 
$\limit(\rho)\geq n{-}1$ then we~know that $\limit(\rho)$ is~even.
Analogously, if~$n$ is~even then $n{-}1$ is~odd and the fact
that $\limit(\rho)$ is~even guarantees that~$\limit(\rho)<n{-}1$.

Clearly,  the above equivalences imply that, under the assumption of~the lemma,
a~strategy winning for the
condition $ W_P$ is~winning for the respective conditions and vice\=/versa.
\end{proof}

Our aim now is~to~define a~function $\fun{\Delta}{\powerset(Q)\times
  A\times\powerset(Q)}{\powerset(Q)}$ that 
%represents strategies of \eve in parts of the games $\CG(t,q))$ over
%a~single letter. 
will allow us~to~form equations over $Q$\=/families.
An~ordered pair of~sets of~states $P_\dL,P_\dR\in\powerset(Q)$ induces
a~\emph{valuation} $v_{P_\dL,P_\dR}$ to~the atoms in~$\{\dL,\dR\}\times Q$ defined by: $v_{P_\dL,P_\dR}(d,p)$ is~true if~$p
\in P_d$. Now, consider additionally a~letter $a\in A$ and put
\[\Delta(P_\dL, a, P_\dR)=\big\{q\in Q\mid v_{P_\dL,P_\dR}\models \delta(q,a) \big\}.\]
%Intuitively, 
Equivalently,
$q\in\Delta(P_\dL,a,P_\dR)$ if~$\eve$ can play the finite game
represented by~$\delta(q,a)$ in~such a~way to~reach only such atoms~$(d,p)$ that
satisfy~$p \in P_d$.

\begin{lemma}
\label{lem:Delta-mono}
The function~$\fun{\Delta}{\powerset(Q)\times A\times\powerset(Q)}{\powerset(Q)}$ is~monotone, i.e.~if $P_\dL\subseteq P_\dL'$ and $P_\dR\subseteq P_\dR'$ then $\Delta(P_\dL, a, P_\dR)\subseteq \Delta(P_\dL',a,P_\dR')$.
\end{lemma}

\begin{proof}
It~follows directly from the fact that the Boolean formulae in~$\delta(q,a)$ are positive.
\end{proof}

Recall that $t\restr_v \in\trees_A$ denotes the subtree of~$ t$
induced by~a~node~$ v$, cf.~Section~\ref{sec:basic}. The following lemma shows how to~increase the index~$i$ of~the above $Q$\=/families $\CS^n_i$ and $\CR^n_i$.

\begin{lemma}
\label{lem:ind-for-S-R-Delta}
Take $n\in\N$, $i\in\N$, and a~tree $t\in\trees_A$. Then we~have:
\begin{align*}
\CS^n_{i+1}[t]&= \Delta\big(\CS^n_i[t\restr_\dL], t(\varepsilon), \CS^n_i[t\restr_\dR]\big),\\
\CR^n_{i+1}[t] &= \Delta\big(\CR^n_i[t\restr_\dL], t(\varepsilon), \CR^n_i[t\restr_\dR]\big).
\end{align*}
\end{lemma}

The proof of~this lemma is~based on~a~standard technique of~merging strategies: the game $\CG(t,p)$ can be~split into a~finite game corresponding to~the formula $\delta\big(p,t(\varepsilon)\big)$ that leads to~the sub\=/games $\CG(t\restr_\dL,p_\dL)$ and~$\CG(t\restr_\dR,p_\dR)$ for some states~$p_\dL,p_\dR\in Q$.

\begin{proof}
Take a~play $\pi$ in the arena $\CG(t,p)$ for some state $p\in Q$. 
Recall that, by the definition of the game, the initial position of the play is $(p, \varepsilon)$ and the next state position will
have the form $(q,d) $, for some $ q \in Q$ and $d\in\{\dL,\dR\}$.
Consider the suffix of the play $ \pi$ starting from that position. 
Clearly, this suffix induces a~play, say $\pi' $,  in the arena
$\CG(t\restr_d, q)$, starting from the position $(q, \varepsilon ) $
(technically,  to satisfy our definition,  we need also to replace
every position $(\alpha , d v) $  by $(\alpha , v) $ in
the original   play).
Moreover,   the sequence of states visited by $ \pi'$, $\stany \,
(\pi')$,  is a suffix of the
sequence $ \stany \, (\pi )$
obtained by removing just the first element.
By the definition of $S^n_i$ and
$R^n_i$ we have therefore 
\begin{equation}
\label{eq:S-and-R-shift}
\stany (\pi )\in S^n_{i+1}\Longleftrightarrow \stany (\pi') \in
S^n_{i}\text{, and }
\stany (\pi) \in R^n_{i+1}\Longleftrightarrow \stany(\pi') \in R^n_{i}.
\end{equation}

We will now provide the proof for $\CS^n_{i+1}$, the case of
$\CR^n_{i+1}$ is analogous. 
Let $P_\dL$ and $P_\dR$ equal respectively $\CS^n_i [t\restr_\dL]$ and
$\CS^n_i[ t\restr_\dR]$. Put $a = t(\varepsilon)$. Recall that by 
the duality of the two representations of $Q$\=/families,  
$p \in  \CS^n_{i+1}[t] $ iff $ t\in\CS^n_{i+1}(p)$.
So we need to prove
that for every $p\in Q$ we have $t\in\CS^n_{i+1}(p)$ if and only if 
$p\in\Delta(P_\dL,a,P_\dR)$.

Assume that $t\in\CS^n_{i+1}(p)$. Take a strategy $\sigma$ witnessing that. Notice that if a position of the form $(q,d)$ can be reached by $\sigma$ then by~\eqref{eq:S-and-R-shift} we know that $t\restr_d\in\CS^n_i(q)$, i.e.~$q\in P_d$. Thus, the strategy $\sigma$ witnesses that $p\in\Delta(P_\dL,a,P_\dR)$.

For the opposite direction, assume that
$p\in\Delta(P_\dL,a,P_\dR)$. This means that there exists a finite
strategy of \eve that allows her to resolve 
the formula $\delta(p,a)$ in such a way that for every atom $(d,q)$
that can be reached by this strategy, we have $(d,q)\in P_d$. 
%This means that \eve has a winning strategy in $\CG(t,p)$ from the
%position $\big(q,d\big)$ with the winning criterion $S^n_i$. 
The last means that \eve has a winning strategy in the game $\CG
(t\restr_d ,q, \CS^n_i ) $.
Now we can combine all above strategies in a~strategy in the game
$\CG(t,p,S^n_{i+1})$, which 
by~Equation~\eqref{eq:S-and-R-shift} is again winning for \eve.
Hence, $t\in\CS^n_{i+1}(p)$, as desired.
\end{proof}

\begin{comment}
The proof of~this lemma is~based on~a~standard technique of~merging strategies: the game $\CG(t,p)$ can be~split into a~finite game corresponding to~the formula $\delta\big(p,t(\varepsilon)\big)$ that leads to~the sub\=/games $\CG(t\restr_\dL,p_\dL)$ and~$\CG(t\restr_\dR,p_\dR)$ for some states~$p_\dL,p_\dR\in Q$. The crucial observation in~the proof is~as follows, if~$\rho\in\runs$ and $\rho'\in\runs$ satisfies $\rho'(k)=\rho(k{+}1)$ for each $k\in\N$ then $\rho\in S^n_{i+1}\Leftrightarrow \rho'\in S^n_i$ and $\rho\in R^n_{i+1}\Leftrightarrow \rho'\in R^n_i$. A~complete proof of~this lemma is~given in~Appendix~\ref{ap:pro-lem-inf-for-Delta}.
\end{comment}

\section{Measures and distributions}
\label{sec:measures}

\begin{figure}
\centering
\newcommand{\chain}[4]{
\evalFloat{\x}{#1}
\evalFloat{\dy}{#2}
\evalFloat{\ds}{0.8}
\evalFloat{\es}{1.5}
\evalInt{\ll}{#4}

\evalFloat{\s}{1.0}
\foreach \i in {0,...,5} {
	\node[nod, scale=\s*1.0] (u\ll\i) at (\x, \y) {$\Fmu(#3_{\i})$};

	\evalFloat{\y}{\y + \dy * \s}
	\evalFloat{\s}{exp(ln(\s)*\es) * \ds}
}

\evalFloat{\s}{1.0}
\foreach \i in {1,...,5} {
	\evalInt{\j}{\i-1}
	\draw (u\ll\j) edge[arr,{|[scale=\s*0.7]}-{Latex[scale=\s*0.7]}] node[left, scale=\s*0.7] {$\CF$} (u\ll\i);
	\evalFloat{\s}{exp(ln(\s)*\es) * \ds}
}

\evalFloat{\y}{\y + 0.3 * \dy}

\node[nod] (l\ll) at (\x, \y) {$\Fmu(#3_\infty)$};
}

\begin{tikzpicture}
\tikzstyle{nod} = [scale=0.9]
\tikzstyle{dtz} = [scale=1.2]
\tikzstyle{arr} = [draw, |-{Latex}]
\newcommand{\dX}{2.5}
\evalFloat{\y}{0}

\evalFloat{\maxY}{\y}

\chain{0*\dX}{-0.9}{\CS^0}{0}

\evalFloat{\minY}{\y}

\chain{1*\dX}{+0.9}{\CR^1}{1}

\draw (l0) edge[arr] node[above] {$\CQ_{<1}$} (u10);

\chain{2*\dX}{-0.9}{\CS^2}{2}

\draw (l1) edge[arr] node[below] {$\CQ_{\geq 2}$} (u20);

\chain{3*\dX}{+0.9}{\CR^3}{3}

\draw (l2) edge[arr] node[above] {$\CQ_{< 3}$} (u30);

\chain{4*\dX}{-0.9}{\CS^4}{4}

\draw (l3) edge[arr] node[below] {$\CQ_{\geq 4}$} (u40);

\node[nod] (u40) at (5*\dX-0.5, \minY) {};

\node[dtz] at (5*\dX, \maxY) {$\cdots$};
\node[dtz] at (5*\dX, \minY*0.5+\maxY*0.5) {$\cdots$};
\node[dtz] at (5*\dX, \minY) {$\cdots$};

\draw (l4) edge[arr] node[above] {$\CQ_{< 5}$} (u40);
\end{tikzpicture}
\caption{A~schematic presentation of~the relationship between the distributions used in~the proof. The vertical axis represents the order~${\preceq}$, i.e.~$\Fmu(\CS^0_0)\succeq \Fmu(\CS^0_\infty)$. The edges marked $\CF$, $\CQ_{<n}$, and $\CQ_{\geq n}$ represent applications of~the respective operations. The vertical convergence is~understood in~terms of~pointwise limits in~$\R^{\powerset(Q)}$.}
\label{fig:construction}
\end{figure}

%\todo{Review 2: suggested to move Figure~\ref{fig:construction} before Lemma 7}

Following an~approach started in~\cite{przybylko_fo_arxiv},  we~transfer
the problem of~computing measures of~tree languages 
$\lang(\CA, p)$ to~computing a~suitable probability distribution on~the sets of~the automaton states.
We~start with a~general construction.
%introduced in~\cite{przybylko_fo_journal}.
For a~finite set~$X$, consider the set of~probability distributions over~$X$,  $\Fdis{X} \eqdef
\big\{\fun{\alpha}{X}{[0,1]}\mid \sum_{x\in X} \alpha(x)=1\big\}$. Observe that, if~$X$ is~partially
ordered by~a~relation~${\leq}$ then $\Fdis X$ is~partially ordered by~a~relation~${\preceq}$ defined as~follows:
$\alpha\preceq \beta$ if~for every upward\=/closed%
\footnote{That is if $x\leq y$ and $x\in U$ then $y\in U$.}
set~$U\subseteq X$, 
we~have $\sum_{x\in U} \alpha(x)\leq \sum_{x\in U}\beta(x)$.
% Similarly, if $\fun{f}{X}{Y}$ is a~function then by $\fun{\Fdis{f}}{\Fdis{X}}{\Fdis{Y}}$ we denote the function defined as $\Fdis{f}(\alpha)(y)=\sum_{x\in f^{-1}(\{y\})} \alpha(x)$.
In~this article, we~are interested in~$\langle X, {\leq}\rangle$ being the powerset~$\powerset(Q)$ ordered by~inclusion~${\subseteq}$.

\begin{remark}
The relation ${\preceq}$ is~a~partial order on~$\Fdis{X}$ (as~an~intersection of~a~finite family of~partial orders).% The function $\Fdis{f}$ is well\=/defined, i.e.~$\Fdis{f}(\alpha)$ is in fact a probabilistic distribution.
\end{remark}

Every $Q$\=/family $\CL$ for a~weak alternating automaton~$\CA$
induces naturally a~member of~$\Fdis{\powerset(Q)}$, which is~a~distribution~$\Fmu\big(\CL\big)$
defined by \[\Fmu\big(\CL\big)(P)=\mu_0\big\{t\in\trees_A\mid \CL[t]=P\big\}.\]

\noindent
Here $\mu_0$ is~the uniform probability measure on~$\trees_A$. The sets in~consideration are measurable thanks to~Proposition~\ref{pro:measurability}.

Note that if~the language family is~exactly 
$\CL (q) = \lang(\CA, q) $  then
the probability assigned to a~set of~states~$ P$ amounts to~the
probability that a~randomly chosen tree, with respect to~$\mu_0 $,  is~accepted precisely from the states in~$P$.

\begin{lemma}
\label{lem:order-of-families}
If~for each $q\in Q$ we~have $\CL(q)\subseteq \CL'(q)$ then $\Fmu(\CL)\preceq \Fmu(\CL')$ in $\Fdis{\powerset(Q)}$.
\end{lemma}

\begin{comment}
A~direct proof of~this lemma is~given in~Appendix~\ref{ap:pro-order-fam}.
\end{comment}

\begin{proof}
Take any upward\=/closed family $U\subseteq \powerset(Q)$. Then
%Notice that
%\begin{align*}
\begin{eqnarray*}
\sum_{P\in U}\ \Fmu\big(\CL\big)(P)  =\sum_{P\in U}\ \mu_0\big\{t\in\trees_A\mid\CL[t]=P \big\}
=\mu_0\big\{t\in\trees_A\mid \CL[t] \in U\big\}  \leq \\
\leq \mu_0\big\{t\in\trees_A\mid \CL'[t] \in U\big\}
= \sum_{P\in U}\ \mu_0\big\{t\in\trees_A\mid\CL'[t]=P
\big\}=\sum_{P\in U}\ \Fmu\big(\CL'\big)(P),
\begin{comment}
\sum_{P\in U}\ \Fmu\big(\CL\big)(P)& =\sum_{P\in U}\ \mu_0\big\{t\in\trees_A\mid\CL[t]=P \big\}\\
&=\mu_0\big\{t\in\trees_A\mid \CL[t] \in U\big\}\\
&\leq \mu_0\big\{t\in\trees_A\mid \CL'[t] \in U\big\}\\
&= \sum_{P\in U}\ \mu_0\big\{t\in\trees_A\mid\CL'[t]=P
\big\}=\sum_{P\in U}\ \Fmu\big(\CL'\big)(P),
\end{comment}
\end{eqnarray*}
%\end{align*}
where the middle inequality follows from the assumption that $\CL(q)\subseteq \CL'(q)$ and the fact that the family $U$ is upward\=/closed.
\end{proof}

We~now examine  the sequences of~distributions $\Fmu\big(\CS^n_i\big)$, $\Fmu\big(\CR^n_i\big)$, $\Fmu\big(\CS^n_\infty\big)$, and $\Fmu\big(\CR^n_\infty\big)$ arising from the $Q$\=/families introduced in~the previous section.
Our aim is~to~bind them by~equations computable within~$\Fdis{\powerset(Q)}$.
As~an~analogue to~the operation~$\Delta$, we~introduce a~function $\fun{\CF}{\Fdis{\powerset(Q)}}{\Fdis{\powerset(Q)}}$ defined for $\beta\in\Fdis{\powerset(Q)}$ and $P\in \powerset(Q)$ by
\begin{equation}
\label{eq:def-F}
\CF(\beta)(P)= \frac{1}{|A|}\cdot\sum_{(P_\dL,a,P_\dR)\in\Delta^{-1}(P)}\ \beta(P_\dL)\cdot \beta(P_\dR).
\end{equation}
Note that the formula guarantees that $\CF(\beta)$ is~indeed a~probabilistic distribution in~$\Fdis\powerset(Q)$.
The operator~$ \CF$  allows us~to~lift the inductive
definitions of~the $Q$\=/families $\CS^n_{i+1}$ and
$\CR^n_{i+1}$ given by~Lemma~\ref{lem:ind-for-S-R-Delta}, to~their counterparts in~the level
of~probability distributions.

\begin{lemma}
\label{lem:step-apply-F}
For each $n\in\N$ and $i\in\N$ we~have
\[\Fmu\big(\CS^n_{i+1}\big)=\CF\Big(\Fmu\big(\CS^n_i\big)\Big)\text{ and }\Fmu\big(\CR^n_{i+1}\big)=\CF\Big(\Fmu\big(\CR^n_i\big)\Big).\]
\end{lemma}

\begin{proof}
Take $P\in\powerset(Q)$ and observe that
\begin{align*}
\CF\Big(\Fmu\big(\CS^n_i\big)\Big)(P) &\eqext{1} \frac{1}{|A|}\cdot\sum_{(P_\dL,a,P_\dR)\in\Delta^{-1}(P)}\ \Fmu\big(\CS^n_i\big)(P_\dL)\cdot  \Fmu\big(\CS^n_i\big)(P_\dR)\\
&\eqext{2} \sum_{(P_\dL,a,P_\dR)\in\Delta^{-1}(P)}\ \mu_0\big\{t_\dL\mid \CS^n_i[t_\dL]{=}P_\dL\big\}\cdot \frac{1}{|A|}\cdot\mu_0\big\{t_\dR\mid \CS^n_i[t_\dR]{=}P_\dR\big\}\\
&\eqext{3} \sum_{(P_\dL,a,P_\dR)\in\Delta^{-1}(P)}\ \mu_0\big\{t\mid \CS^n_i[t\restr_\dL]{=}P_\dL\land t(\varepsilon){=}a\land \CS^n_i[t\restr_\dR]{=}P_\dR\big\}\\
&\eqext{4}\mu_0\left(\bigcup_{(P_\dL,a,P_\dR)\in\Delta^{-1}(P)}\ \big\{t\mid \CS^n_i[t\restr_\dL]{=}P_\dL\land t(\varepsilon){=}a\land \CS^n_i[t\restr_\dR]{=}P_\dR\big\}\right)\\
&\eqext{5} \mu_0\Big\{t\in\trees_A\mid \Delta\big(\CS^n_i[t\restr_\dL], t(\varepsilon), \CS^n_i[t\restr_\dR]\big){=}P\Big\}\\
&\eqext{6} \mu_0\Big\{t\in\trees_A\mid \CS^n_{i+1}[t]{=}P\Big\} \eqext{7} \Fmu\big(\CS^n_{i+1}\big)(P),
\end{align*}
where:~(1) is~just the definition of
$\CF\Big(\Fmu\big(\CS^n_i\big)\Big)$; (2) follows from the definition
of~$\Fmu\big(  \CS^n_i \big)$; (3) follows from 
Remark~\ref{rem:measure} and the independence of~$\{ t(\varepsilon )=a \}
$  from the other events in~consideration;
(4)~follows from the fact that the measured sets are pairwise
disjoint; (5)~follows 
%from the fact that $\Delta$ is a~function; 
simply from the definition of~$\Delta$;
(6)~follows from Lemma~\ref{lem:ind-for-S-R-Delta}; and (7) is~just the
definition of~$\Fmu \big(\CS^n_{i+1}\big)$.

The proof for $\CR^n_{i+1}$ follows the same steps, except it~uses the $\CR^n_i$ variant of~Lemma~\ref{lem:ind-for-S-R-Delta}.
\end{proof}

Now, recall that $Q_{\geq n}$ and $Q_{<n}$ are sets of~states
of~respective priorities. Let 
the functions~$\fun{\CQ_{<n},\allowbreak\CQ_{\geq n}}{\Fdis\powerset(Q)}{\Fdis\powerset(Q)}$ be~defined by
\begin{eqnarray}
%\begin{align}
\CQ_{< n}(\beta)(P)\eqdef&\!\!\sum_{P'\colon P'\cap Q_{< n}=P}\beta(P'),\label{eq:def-CQ-1}
\\
\CQ_{\geq n}(\beta)(P)\eqdef&\!\!\sum_{P'\colon  P'\cup Q_{\geq n}=P} \beta(P').\label{eq:def-CQ-2}
%\end{align}
\end{eqnarray}
\noindent
Again, the formulae guarantee that $\CQ_{< n}(\beta)$ and $\CQ_{\geq n}(\beta)$ are both probabilistic distributions in~$\Fdis\powerset(Q)$.
The following lemma shows the relation between these functions and the limit distributions $\Fmu\big(\CS^{n-1}_{\infty}\big)$ and $\Fmu\big(\CR^{n-1}_{\infty}\big)$.

\begin{lemma}
\label{lem:compute-apply-s-r}
For each $n\in\N$ we~have
\begin{align*}
\CQ_{< n}\Big(\Fmu\big(\CS^{n-1}_\infty\big)\Big)&=\Fmu\big(\CR^n_0\big)&\text{if $n$ is odd,}\\
\CQ_{\geq n}\Big(\Fmu\big(\CR^{n-1}_\infty\big)\Big)&=\Fmu\big(\CS^n_0\big)&\text{if $n$ is even.}
\end{align*}
\end{lemma}
This lemma follows from Lemma~\ref{lem:S-R-limit-steps} in~a~similar way as~Lemma~\ref{lem:step-apply-F} follows from Lemma~\ref{lem:ind-for-S-R-Delta}.

\begin{proof}
Consider the case of even $n$ and a~tree $t\in\trees_A$. We need to show that
\[\CQ_{\geq n}\Big(\Fmu\big(\CR^{n-1}_\infty\big)\Big)=\Fmu\big(\CS^n_0\big).\]
Lemma~\ref{lem:S-R-limit-steps} implies that 
\begin{equation}
\CS^n_0[t] = \big(\CR^{n-1}_\infty[t]\big)\cup Q_{\geq n}.
\end{equation}
Therefore, for each $P\in\powerset(Q)$ we have
\begin{align*}
\Fmu\big(\CS^n_0\big)(P)&=\mu_0\big\{t\in\trees_A\mid \CS^n_0[t] = P\big\}\\
&=\mu_0\big\{t\in\trees_A\mid \big(\CR^{n-1}_\infty[t]\big)\cup Q_{\geq n} = P\big\}\\
&=\mu_0\left(\bigcup_{P'\colon P'\cup Q_{\geq n}=P}\ \big\{t\in\trees_A\mid\CR^{n-1}_\infty[t]=P'\big\}\right)\\
&=\sum_{P'\colon P'\cup Q_{\geq n}=P}\ \mu_0\big\{t\in\trees_A\mid\CR^{n-1}_\infty[t]=P'\big\}\\
&=\sum_{P'\colon P'\cup Q_{\geq n}=P}\ \Fmu\big(\CR^{n-1}_\infty\big)(P')\\
&= \CQ_{\geq n}\big(\Fmu\big(\CR^{n-1}_\infty\big)\big)(P)
\end{align*}
The case of odd $n$ is entirely analogous.
\end{proof}

The two above lemmata express the properties of~the operators $\CF$, $\CQ_{< n}$, and $\CQ_{\geq n}$ as~depicted on~Figure~\ref{fig:construction}.

%are enough to perform the following computations in $\Fdis\powerset(Q)$:
%\begin{align*}
%\CF\colon\ \quad\, \Fmu\big(\CS^n_{i}\big)\ &\longmapsto\ \Fmu\big(\CS^n_{i+1}\big),\\
%\CF\colon\ \quad \Fmu\big(\CR^n_{i}\big)\ &\longmapsto\ \Fmu\big(\CR^n_{i+1}\big),\\
%\CQ_{< n}\colon\ \, \Fmu\big(\CS^{n-1}_\infty\big)\  &\longmapsto\ \Fmu\big(\CR^n_0\big)&\text{for $n$ odd,}\\
%\CQ_{\geq n}\colon\ \Fmu\big(\CR^{n-1}_\infty\big)\ &\longmapsto\ \Fmu\big(\CS^n_0\big)&\text{for $n$ even.}
%\end{align*}
%Therefore, the only distributions for which we need to~provide a~computable definitions 
%are the limit distributions~$\Fmu\big(\CS^{n}_\infty\big)$ 
%and $\Fmu\big(\CR^{n}_\infty\big)$. To~accomplish this task, we~connect 
%the two views of~$\Fdis \powerset(Q)$: as~an~ordered set with respect to~${\preceq}$,  and as~a~subspace of~the
%topological space~$\R^{\powerset(Q)}$.

\section{Limit distributions \texorpdfstring{$\Fmu\big(\CS^n_\infty\big)$}{mu(Sin)} and \texorpdfstring{$\Fmu\big(\CR^n_\infty\big)$}{mu(Rin)}}
\label{sec:limit}

%The above section shows how to go from $\Fmu\big(\CS^n_i\big)$ to $\Fmu\big(\CS^n_{i+1}\big)$ and how to obtain $\Fmu\big(\CS^n_0\big)$ from $\Fmu\big(\CR^{n-1}_\infty\big)$ (and dually for $\CS$) within $\Fdis\powerset(Q)$. Thus, the only thing we lack is a way to compute $\Fmu\big(\CS^n_\infty\big)$ in an effective way.
In~this section we show how to represent the distributions~$\Fmu\big(\CS^n_\infty\big)$ and
$\Fmu\big(\CR^n_\infty\big)$ 
as~fixed points. We~begin by~proving that 
%the distributions~$\Fmu\big(\CS^n_\infty\big)$ and $\Fmu\big(\CR^n_\infty\big)$ 
these distributions
are limits in~$\R^{\powerset(Q)}$ of~the sequences of~vectors $\big(\Fmu\big(\CS^n_i\big)\big)_{i\in\N}$ and $\big(\Fmu\big(\CR^n_i\big)\big)_{i\in\N}$ respectively. This 
is~a~consequence  
% stems 
of~Lemmata~\ref{lem:S-R-monotone} and~\ref{lem:limit-steps-S-R}.

\begin{lemma}
\label{granica}
For each $n\in\N$ and $P\in\powerset(Q)$ we~have
\[\lim_{i\to\infty} \Fmu\big(\CS^n_i\big)(P)=\Fmu\big(\CS^n_\infty\big)(P)\text{ and }\lim_{i\to\infty} \Fmu\big(\CR^n_i\big)(P)=\Fmu\big(\CR^n_\infty\big)(P).\]
\end{lemma}

\begin{proof}
We~consider case of~$\Fmu\big(\CS^n_\infty\big)$, the case of~$\Fmu\big(\CR^n_\infty\big)(P)$ is~entirely dual.
First, we show that the respective limits agree when taking sums over any upward closed family $U\subseteq\powerset(Q)$, see~\eqref{eq:upward-agrees} below. For $i\in \N$ let $X_i=\bigcup_{P'\in U}\{t\in\trees_A\mid \CS^n_i[t]=P'\}$ and $X_\infty = \bigcup_{P'\in U}\{t\in\trees_A\mid \CS^n_\infty[t]=P'\}$. Lemma~\ref{lem:S-R-monotone} together with the fact that $U$ is~upward\=/closed imply that $X_0\supseteq X_1\supseteq\ldots\supseteq X_\infty$. Lemma~\ref{lem:limit-steps-S-R} and finiteness of~$Q$ imply that for every tree~$t$ there exists an~index~$J$ such that $\CS^n_J[t]\subseteq \CS^n_\infty[t]$. Therefore, $\bigcap_{i\in \N} X_i= X_\infty$. By continuity of~the measure~$\mu_0$ we get that $\lim_{i\to\infty} \mu_0(X_i)=\mu_0(X_\infty)$. This means that
\begin{equation}
\label{eq:upward-agrees}
\lim_{i\to\infty} \sum_{P'\in U}\Fmu\big(\CS^n_i\big)(P')=\lim_{i\to\infty} \mu_0(X_i) = \mu_0(X_\infty)= \sum_{P'\in U}\Fmu\big(\CS^n_\infty\big)(P').
\end{equation}

Clearly, $\{P\}=\{P'\in\powerset(Q)\mid P'\supseteq P\}\setminus \{P'\in\powerset(Q)\mid P'\supsetneq P\}$ with both these families upward closed. Therefore, we~can apply~\eqref{eq:upward-agrees} twice and obtain the desired equation.
%
%Let
%Follows directly from Lemmata~\ref{lem:S-R-monotone} and
%~\ref{lem:limit-steps-S-R} and continuity of measure: if
%$(X_i)_{i\in\N}$ is an increasing or decreasing sequence of sets, 
%then $\lim_{i\to\infty} \mu_0(X_i) $ equals
%$\mu_0\big(\bigcup_{i\in\N} X_i\big)$ or $\mu_0\big(\bigcap_{i\in\N}
%X_i\big)$, respectively.
\end{proof}

%The following corollary follows directly from Lemmata~\ref{lem:order-of-families} and~\ref{lem:S-R-monotone}.

%\begin{corollary}
%\label{cor:mus-are-monotone}
%For each $n\in\N$ and $i\in\N$ we have
%\[\Fmu\big(\CS^n_i\big)\succeq \Fmu\big(\CS^n_{i+1}\big)\text{ and }\Fmu\big(\CR^n_i\big)\preceq \Fmu\big(\CR^n_{i+1}\big).\]
%\end{corollary}

The monotonicity of~$\Delta$, see Lemma~\ref{lem:Delta-mono}, implies the following lemma.

\begin{lemma}
\label{lem:monotone-F}
The operator~$\fun{\CF}{\Fdis{\powerset(Q)}}{\Fdis{\powerset(Q)}}$, see~Equation~\eqref{eq:def-F}, is~monotone in~${\preceq}$.
\end{lemma}

\begin{proof}
%\[\CF(\beta)(P)= \frac{1}{|A|}\cdot\sum_{(P_\dL,a,P_\dR)\in\Delta^{-1}(P)}\ \beta(P_\dL)\cdot \beta(P_\dR).\]
We need to prove that $\CF$ is monotone w.r.t. the order ${\preceq}$. Thus, for every $\alpha\preceq \beta\in\Fdis{\powerset(Q)}$ and an~upward\=/closed family $U\subseteq \powerset(Q)$ we should have $\sum_{P\in U} \CF(\alpha)(P)\leq \sum_{P\in U} \CF(\beta)(P)$.

After multiplying by $\frac{1}{|A|}$ and splitting the sum over separate letters $a\in A$ (see the definition of $\CF$, cf.~\eqref{eq:def-F}), it is enough to show that for each $a\in A$ and $O_a\eqdef \{(P_\dL,P_\dR)\mid\Delta(P_\dL,a,P_\dR)\in U\}$ we have
\[\sum_{(P_\dL, P_\dR)\in O_a}\alpha(P_\dL)\cdot \alpha(P_\dR)\leq\sum_{(P_\dL, P_\dR)\in O_a}\beta(P_\dL)\cdot \beta(P_\dR).\]
Now, by monotonicity of $\Delta$ (see Lemma~\ref{lem:Delta-mono}) and the fact that $U$ is upward\=/closed, we know that if $P_\dL\subseteq P'_\dL$, $P_\dR\subseteq P'_\dR$, and $(P_\dL,P_\dR)\in O_a$ then also $(P'_\dL,P'_\dR)\in O_a$. By $P_{\dL}^{-1}\cdot O_a$ and $O_a\cdot P_\dR^{-1}$ we will denote the sections $\{P_\dR\mid (P_\dL,P_\dR)\in O_a\}$ and $\{P_\dL\mid (P_\dL,P_\dR)\in O_a\}$ respectively. Notice that both of them are upward\=/closed. Thus, using the assumption that $\alpha\preceq \beta$ twice, we obtain
\begin{align*}
\sum_{(P_\dL, P_\dR)\in O_a}\alpha(P_\dL)\cdot \alpha(P_\dR)&= \sum_{P_\dL\in\powerset(Q)} \alpha(P_\dL)\cdot \left(\sum_{P_\dR \in P_\dL^{-1}\cdot O_a} \alpha(P_\dR)\right)\\
&\leq \sum_{P_\dL\in\powerset(Q)} \alpha(P_\dL)\cdot \left(\sum_{P_\dR \in P_\dL^{-1}\cdot O_a} \beta(P_\dR)\right)\\
&=\sum_{(P_\dL, P_\dR)\in O_a}\alpha(P_\dL)\cdot \beta(P_\dR)=\sum_{(P_\dL, P_\dR)\in O_a}\beta(P_\dR)\cdot\alpha(P_\dL)\\
&=\sum_{P_\dR\in\powerset(Q)} \beta(P_\dR)\cdot \left(\sum_{P_\dL \in O_a\cdot P_\dR^{-1}} \alpha(P_\dL)\right)\\
&\leq\sum_{P_\dR\in\powerset(Q)} \beta(P_\dR)\cdot \left(\sum_{P_\dL \in O_a\cdot P_\dR^{-1}} \beta(P_\dL)\right)\\
&=\sum_{(P_\dL, P_\dR)\in O_a}\beta(P_\dR)\cdot \beta(P_\dL)=\sum_{(P_\dL, P_\dR)\in O_a}\beta(P_\dL)\cdot\beta(P_\dR).
\end{align*}
\end{proof}

\begin{comment}
\noindent
A~detailed proof of~the lemma is~provided in~Appendix~\ref{ap:pro-F-mono}.
\end{comment}

With the two lemmata above, we~are ready to~conclude this section:
we characterise the distributions $\Fmu\big(\CS^n_\infty\big)$ and $\Fmu\big(\CR^n_\infty\big)$, see Figure~\ref{fig:construction}, by~a~specialised
variant of~the Knaster\=/Tarski fixed point theorem.

%The following corollary follows from Lemma~\ref{lem:def-of-fixpoint}, using Lemmata~\ref{lem:step-apply-F} and~\ref{lem:monotone-F} and Corollary~\ref{cor:mus-are-monotone} to check the respective properties of $\Fmu\big(\CS^n_\infty\big)$ and $\Fmu\big(\CR^n_\infty\big)$.

\begin{proposition}
\label{pro:compute-limits}
For each $n\in\N$ the distribution $\Fmu\big(\CS^n_\infty\big)$ is the ${\preceq}$\=/greatest distribution~$\beta$ satisfying $\beta=\CF(\beta)$ and $\beta \preceq \Fmu\big(\CS^n_0\big)$.
Similarly, $\Fmu\big(\CR^n_\infty\big)$ is the ${\preceq}$\=/least distribution~$\beta $ satisfying $\beta=\CF(\beta)$ and $\beta \succeq \Fmu\big(\CR^n_0\big)$.
\end{proposition}

\begin{proof}
Consider the case of~$\Fmu\big(\CS^n_\infty\big)$. Observe that $\CF$ is~continuous in~$\R^{\powerset(Q)}$. Indeed, it~is~given by a~vector of~quadratic polynomials from $\R^{\powerset(Q)}$ to $\R^{\powerset(Q)}$. Now, take $P\in\powerset(Q)$ and observe that
\begin{align*}
\Fmu\big(\CS^n_\infty\big)(P)=
\lim_{i\to\infty} \Fmu\big(\CS^n_i\big)(P)=
&\lim_{i\to\infty} \CF\big(\Fmu\big(\CS^n_i\big)\big)(P)=\\
&\CF\Big(\lim_{i\to\infty} \Fmu\big(\CS^n_i\big)(P)\Big)=
\CF\Big(\Fmu\big(\CS^n_\infty\big)(P)\Big)
\end{align*}
where the first equality follows from Lemma~\ref{granica}; the second from Lemma~\ref{lem:step-apply-F}; the third from continuity of~$\CF$; and the last from Lemma~\ref{granica}, again. Therefore, $\beta= \Fmu\big(\CS^n_\infty\big)$ satisfies $\beta=\CF(\beta)$. Moreover, Lemmata~\ref{lem:S-R-monotone}
and~\ref{lem:order-of-families} imply that $\beta\preceq \Fmu\big(\CS^n_0\big)$.

Consider now any distribution $\beta\in\Fdis{\powerset(Q)}$ such that $\beta=\CF(\beta)$ and $\beta\preceq \Fmu\big(\CS^n_0\big)$. We~need to prove that $\beta\preceq\Fmu\big(\CS^n_\infty\big)$. Lemma~\ref{lem:monotone-F} states that $\CF$ is~monotone. Therefore, by~inductively applying Lemma~\ref{lem:step-apply-F} for $i=0,\ldots$, we~infer that
\[\beta=\CF(\beta)\leq \CF\Big(\Fmu\big(\CS^n_i\big)\Big)=\Fmu\big(\CS^n_{i+1}\big).\]
Take any upward\=/closed family $U\subseteq\powerset(Q)$. The above inequality implies that for each $i\in\N$ we~have $\sum_{P\in U} \beta(P)\leq \sum_{P\in U} \Fmu\big(\CS^n_i\big)(P)$. By~taking the limit as~in~Lemma~\ref{granica} we~obtain that $\sum_{P\in U} \beta(P)\leq \sum_{P\in U} \Fmu\big(\CS^n_\infty\big)(P)$. This implies that $\beta\preceq \Fmu\big(\CS^n_\infty\big)$.

The case of~$\Fmu\big(\CR^n_\infty\big)$ is~similar, we utilise the opposite monotonicity $\beta\succeq \Fmu\big(\CR^n_{i+1}\big)$.
\end{proof}
%The above proposition provides an~implicit definition of~the limit
%distributions~$\Fmu\big(\CS^n_\infty\big)$ and
%$\Fmu\big(\CR^n_\infty\big)$. 
%In~the next section, we~describe how to~actually compute these distributions.

%\section{Computations in \texorpdfstring{$\Fdis X$}{DX}}
\section{Computing measures}
\label{sec:computations}

In~this section, we~conclude our solution to~Problem~\ref{pro:comp-mso} for weak alternating automata. This is~achieved by~a~reduction to~the 
first\=/order theory of~the real numbers ${\cal R} = \langle \R, 0,1,{+}, {\cdot}\rangle$.
The theory is~famously decidable thanks to~Tarski\=/Seidenberg  theorem, see e.g.~\cite{tarski_decision}. 
%To~be~more precise, for a~given weak alternating automaton~$\CA$ we~construct a~first\=/order formula that uniquely describes the number $\mu_0\big(\lang(\CA)\big)$; we~justify this approach further in~Subsection~\ref{ssec:algebra} below.

Throughout this section, we assume that the reader is~familiar with the syntax and semantics of~first\=/order logic.  
We~say that a~formula~$\varphi (x_1, \ldots ,x_k) $ {\em represents} a~relation~$ r \subseteq \R^k$ if~it~holds in~$\R$ according to~an~evaluation~$v$ of~the free  variables~$ x_1,\ldots ,x_n$,  precisely when the tuple~$\langle v(x_1),\ldots , v(x_k)\rangle  $  belongs to~$ r$.  For example, the formula~$ \exists z.\ x + (z {\cdot} z) = y$ represents the standard ordering~${\leq}$  on~real numbers.  A~formula represents a~number~$ a \in \R$ if~it~represents the singleton~$\{ a\} $; for example the formula~$ x {\cdot} x = 1 {+} 1 \land \exists z.\ x = z {\cdot} z$, represents the number~$\sqrt{2}$.

%Due to decidability results~\cite{tarski_decision,collins_algebraic_decomposition,ben_tarski_complexity} we treat $\psi_\CA(x)$ as a~representation of $\mu_0\big(\lang(\CA)\big)$, see Subsection~\ref{ssec:algebra} below.

%Recall that \emph{first\=/order logic} is a~formal language defined by the grammar:
%\[\psi::= \lnot\psi\vbar \psi\lor \psi'\vbar \psi\land \psi'\vbar\exists x.\ \psi(x)\vbar\forall x.\ \psi(x)\vbar R(x_1,\ldots,x_n),\]
%where $R$ is a~relational symbol from the given signature $\Sigma$. In this article we fix the signature of real numbers, i.e.~of the structure $\R=\langle \R, 0,1,{+}, {\cdot}\rangle$.

\begin{theorem}
\label{thm:exists-psi}
Given a~weak alternating automaton $\CA$ one can compute a~formula $\psi_\CA(x)$ that 
represents the number  $ \mu_0\big(\lang(\CA)\big) $.
%holds in $\R$ for a~unique number $x=\mu_0\big(\lang(\CA)\big)$. 
Moreover, the formula is in a~prenex normal form, its size is exponential in the size of $\CA$, and the quantifier alternation of $\psi_\CA(x)$ is constant.
\end{theorem}

\begin{proof}
Fix a~weak alternating automaton~$\CA=\langle A, Q, q_\init, \delta, \Omega\rangle$. Let $N>\Omega(q_\init)$ be an~even number (either $\Omega(q_\init){+}1$ or $\Omega(q_\init){+}2$). Fix an~enumeration $\{P_1,\ldots,P_K\}$ of $\powerset(Q)$ with $K=2^{|Q|}$. We will identify a~distribution $\alpha\in\Fdis{\powerset(Q)}$ with its representation $\alpha=(a_1,\ldots,a_K)\in\R^K$ as a~vector of real numbers. Following this identification, $\alpha(P_k)$ stands for $a_k$. Clearly the properties that $\CF(\alpha)=\beta$, $\CQ_{< n}(\alpha)=\beta$, and $\CQ_{\geq n}(\alpha)=\beta$ are definable by quantifier free formulae of size polynomial in~$K$.

The following formula defines the fact that  $\alpha\in\Fdis{\powerset(Q)}$.
\begin{align*}
\mathrm{dist}(\alpha)\equiv \sum_{k=1}^K\alpha(P_k)=1\ \land\ \bigwedge_{k=1}^K 0\leq \alpha(P_k)\leq 1
\end{align*}

Analogously to our representation of distributions, every 
%upward closed set 
subset
$U\subseteq\powerset(Q)$ can be represented by its indicator: a~vector of numbers $\iota=(i_1,\ldots,i_K)$ such that $\iota(P)$ is either~$0$ (if $P\notin U$) or~$1$ (if $P\in U$).  Note that if $ U$ is upward closed then whenever $P\subseteq P'$ and $\iota(P)=1$ then $\iota(P')=1$.  The following formula defines the fact that $ \iota$  represents an upward-closed set.
\begin{align*}
\mathrm{upward}(\iota)\equiv \bigwedge_{k=1}^K \big(\iota(P_k){=}0\lor\iota(P_k){=}1\big)\land \bigwedge_{P\subseteq P'} \iota(P){=}1\rightarrow \iota(P'){=}1.
\end{align*}

Thus, to check if $\alpha\preceq \beta$ one can use the following formula (see Claim~\ref{cl:order-by-minor} below)
\begin{align*}
\mathrm{minor}(\alpha,\beta,\iota)\equiv \sum_{k=1}^K \alpha(P_k)\cdot \iota(P_k)\leq \sum_{k=1}^K \beta(P_k)\cdot \iota(P_k).
\end{align*}

Notice that all the above formulae: $\mathrm{dist}(\alpha)$, $\mathrm{upward}(\iota)$, and $\mathrm{minor}(\alpha,\beta,\iota)$ are quantifier free: the $\bigwedge$ there are just explicitly written as finite conjunctions. Therefore, these formulae can be used to relativise quantifiers in a~prenex normal form of a~formula: for instance we write $\forall \alpha\colon \mathrm{dist}(\alpha).\ \exists \beta\colon \mathrm{dist}(\beta).\ \psi(\alpha,\beta)$ to denote $\forall \alpha.\ \exists \beta.\ \mathrm{dist}(\alpha)\rightarrow\big(\mathrm{dist}(\beta)\land\psi(\alpha,\beta)\big)$.

\begin{claim}
\label{cl:order-by-minor}
Given two distributions $\alpha$ and $\beta$, we have $\alpha\preceq\beta$ if and only if
\[\forall \iota\colon \mathrm{upward}(\iota).\ \mathrm{minor}(\alpha,\beta,\iota).\]
\end{claim}

\begin{figure}
\centering
\begin{tikzpicture}
\tikzstyle{nod} = [scale=0.9]
\tikzstyle{dtz} = [scale=1.2]
\tikzstyle{arr} = [draw, -{Latex}]
\evalFloat{\x}{3}

\node[nod] (u0) at (0*\x, +1) {$\alpha_0=\Fmu(\CS^0_0)$};
\node[nod] (l0) at (0*\x, -1) {$\beta_0=\Fmu(\CS^0_\infty)$};

\draw (u0) edge[arr] node[left] {$\heartsuit$} (l0);

\node[nod] (u1) at (1*\x, +1) {$\beta_1=\Fmu(\CR^0_\infty)$};
\node[nod] (l1) at (1*\x, -1) {$\alpha_1=\Fmu(\CR^0_0)$};

\draw (l1) edge[arr] node[left] {$\spadesuit$} (u1);

\draw (l0) edge[arr] node[above] {$\CQ_{<1}$} (l1);

% -----

\node[nod] (u2) at (2*\x, +1) {$\alpha_2=\Fmu(\CS^0_0)$};
\node[nod] (l2) at (2*\x, -1) {$\beta_2=\Fmu(\CS^0_\infty)$};

\draw (u2) edge[arr] node[left] {$\heartsuit$} (l2);

\node[nod] (u3) at (3*\x, +1) {$\beta_3=\Fmu(\CR^0_\infty)$};
\node[nod] (l3) at (3*\x, -1) {$\alpha_3=\Fmu(\CR^0_0)$};

\draw (l3) edge[arr] node[left] {$\spadesuit$} (u3);

\draw (l2) edge[arr] node[above] {$\CQ_{<3}$} (l3);

\draw (u1) edge[arr] node[below] {$\CQ_{\geq 2}$} (u2);
% ----

\node[nod] (u4) at (4*\x-0.75, +1) {};

\node[dtz] at (4*\x, +1) {$\cdots$};
\node[dtz] at (4*\x, +0) {$\cdots$};
\node[dtz] at (4*\x, -1) {$\cdots$};

\draw (u3) edge[arr] node[below] {$\CQ_{\geq 4}$} (u4);
\end{tikzpicture}
\caption{A diagram of the distributions $\alpha_n$ and $\beta_n$ in the formula $\psi_\CA(x)$, cf.~Figure~\ref{fig:construction}. The symbol $\heartsuit$ represents applications of Proposition~\ref{pro:compute-limits} in the case of $\CS^n_0$ and $\CS^n_\infty$; while $\spadesuit$ corresponds to the dual case of $\CR^n_0$ and $\CR^n_\infty$.}
\label{fig:how-ab-are-chosen}
\end{figure}

The formula $\psi_\CA(x)$ is indented to specify the distributions $(\alpha_n,\beta_n)_{n=0,\ldots,N}$ in a~way depicted on Figure~\ref{fig:how-ab-are-chosen}. The value $\Fmu\big(\CS^0_0(P))$ is~$1$ if~$P=Q$ and $0$ otherwise, see Lemma~\ref{lem:S-R-limit-steps}.
Proposition~\ref{pro:compute-limits} allows us~to~define $\Fmu\big(\CS^n_\infty)$ (resp.~$\Fmu\big(\CR^n_\infty)$) using $\Fmu\big(\CS^n_0)$ (resp.~$\Fmu\big(\CR^n_0)$) as specific fixed points of the operation~$\CF$.
Finally, Lemma~\ref{lem:compute-apply-s-r} allows us~to~define $\Fmu\big(\CR^n_0\big)$ using $\CQ_{<n}$, and $\Fmu\big(\CS^n_0\big)$ using $\CQ_{\geq n}$. The value of $x$ is related to those distributions based on~Lemma~\ref{lem:S-R-limit-steps} which implies that $\mu_0\big(\lang(\CA)\big)=\sum_{P\colon q_\init\in P\in\powerset(Q)}\ \Fmu\big(\CS^N_0\big)(P)$.

The following equation defines the formula $\psi_\CA(x)$.
\begingroup
\allowdisplaybreaks
\begin{align}
\psi_\CA(x)\equiv&\ \exists \alpha_0,\beta_0\colon \mathrm{dist}(\alpha_0), \mathrm{dist}(\beta_0), \beta_0 {=} \CF(\beta_0).\nonumber\\
&\ \quad\vdots\label{eq:it-ex-ab}\\
&\ \exists \alpha_N,\beta_N\colon \mathrm{dist}(\alpha_N), \mathrm{dist}(\beta_N), \beta_N {=} \CF(\beta_N).\nonumber\\
&\ \quad \forall \theta\colon \mathrm{dist}(\theta), \theta {=} \CF(\theta). \label{eq:it-ex-t}\\
&\ \qquad \exists \iota_0\colon \mathrm{upward}(\iota_0).\nonumber\\
&\ \qquad \quad\vdots\label{eq:it-ex-i}\\
&\ \qquad \exists \iota_N\colon \mathrm{upward}(\iota_N).\nonumber\\
&\ \qquad \quad \forall \gamma_0\colon \mathrm{upward}(\gamma_0).\nonumber\\
&\ \qquad \quad \quad\vdots\label{eq:it-fa-t}\\
&\ \qquad \quad \forall \gamma_N\colon \mathrm{upward}(\gamma_N).\nonumber\\
&\ \qquad \qquad \left(\alpha_0(Q)=1\land\bigwedge_{P\neq Q} \alpha_0(P)=0\right)\ \land\label{eq:it-alphaO}\\
&\ \qquad \qquad \left(\bigwedge_{n=1}^N [\text{$n$ is odd}]\rightarrow \alpha_n = \CQ_{<n}(\beta_{n-1})\right)\ \land\label{eq:it-b-to-a-o}\\
&\ \qquad \qquad \left(\bigwedge_{n=1}^N [\text{$n$ is even}]\rightarrow \alpha_n = \CQ_{\geq n}(\beta_{n-1})\right)\ \land\\
&\ \qquad \qquad \left(\bigwedge_{n=0}^N [\text{$n$ is odd}]\rightarrow \mathrm{minor}(\alpha_n,\beta_n,\gamma_n)\right)\ \land\label{eq:it-order-ab-o}\\
&\ \qquad \qquad \left(\bigwedge_{n=0}^N [\text{$n$ is even}]\rightarrow \mathrm{minor}(\beta_n,\alpha_n,\gamma_n)\right)\ \land\label{eq:it-order-ab-e}\\
&\ \qquad \qquad \left(\bigwedge_{n=0}^N [\text{$n$ is odd}]\rightarrow\Big(\lnot\mathrm{minor}(\alpha_n,\theta,\iota_n)\lor \mathrm{minor}(\beta_n,\theta,\gamma_n)\Big)\right)\ \land\label{eq:it-order-bt-o}\\
&\ \qquad \qquad \left(\bigwedge_{n=0}^N [\text{$n$ is even}]\rightarrow\Big(\lnot\mathrm{minor}(\theta,\alpha_n,\iota_n)\lor \mathrm{minor}(\theta,\beta_n,\gamma_n)\Big)\right)\ \land\label{eq:it-order-bt-e}\\
&\ \qquad \qquad \left(\sum_{P\ni q_\init} \alpha_N(P)=x\right)\label{eq:it-x}
\end{align}
\endgroup

Observe that the size of this formula is polynomial in $K$ and $N$ (in fact it is $\CO(N\cdot K^2)$), i.e.~exponential in the size of the automaton $\CA$. Moreover, the formula is in prenex normal form and its quantifier alternation is $4$ (the sub\=/formulae that involve $\bigwedge$ are written explicitly as conjunctions).

We begin by proving soundness of the formula: we assume that $\psi_\CA(x)$ holds and show that $x=\mu_0\big(\lang(\CA)\big)$. Consider a~sequence of distributions $(\alpha_n,\beta_n)_{n=0,\ldots,N}$ witnessing~\eqref{eq:it-ex-ab}. The following two lemmata prove inductively that for $n=0,\ldots,N$ we have
\begin{align}
\alpha_n = \Fmu(\CS^n_0)&\text{ and }\beta_n = \Fmu(\CS^n_\infty)&\text{for even $n$,}\label{eq:req-for-ab}\\
\alpha_n = \Fmu(\CR^n_0)&\text{ and }\beta_n = \Fmu(\CR^n_\infty)&\text{for odd $n$.}\nonumber
\end{align}

\begin{lemma}
\label{lem:ind-step-a-to-b}
Using the above notations and the assumption that $\psi_\CA(x)$ holds:
\begin{align*}
\text{for even $n$, if }\alpha_n = \Fmu(\CS^n_0)&\text{ then }\beta_n = \Fmu(\CS^n_\infty),\\
\text{for odd $n$, if }\alpha_n = \Fmu(\CR^n_0)&\text{ then }\beta_n = \Fmu(\CR^n_\infty).
\end{align*}
\end{lemma}

\begin{proof}
Both claims follow from Proposition~\ref{pro:compute-limits}. Take $n$ odd and assume that $\alpha_n=\Fmu(\CR^n_0)$. We know that $\beta_n=\CF(\beta_n)$ by~\eqref{eq:it-ex-ab}. Moreover, by Claim~\ref{cl:order-by-minor}, the arbitrary choice of $\gamma_n$, and~\eqref{eq:it-order-ab-o} we know that $\alpha_n\preceq \beta_n$. It is enough to prove that if $\theta$ is any distribution satisfying $\alpha_n\preceq \theta$ and $\theta=\CF(\theta)$ then $\beta_n\preceq\theta$.

Assume contrarily that $\theta$ is a~distribution such that $\alpha_n\preceq \theta$ and $\theta=\CF(\theta)$ but $\beta_n\npreceq\theta$. We know that $\theta$ must satisfy the sub\=/formula in~\eqref{eq:it-ex-t}. Take the upward closed sets $(\iota_\ell)_{\ell=0,\ldots,N}$ given by~\eqref{eq:it-ex-i}. Now let $(\gamma_\ell)_{\ell=0,\ldots,N}$ be any sequence of upward closed sets such that $\gamma_n$ witnesses the fact that $\beta_n\npreceq\theta$, i.e.~$\lnot\mathrm{minor}(\beta_n,\theta,\gamma_n)$ holds. But this is a~contradiction with~\eqref{eq:it-order-bt-o} because $\mathrm{minor}(\alpha_n,\theta,\iota_n)$ is true as $\alpha_n\preceq\theta$ and $\mathrm{minor}(\beta_n,\theta,\gamma_n)$ is false.

The case of even $n$ is analogous.
\end{proof}

%The following lemma together with Lemma~\ref{lem:ind-step-a-to-b} conclude the proof of~\eqref{eq:req-for-ab}.

\begin{lemma}
Using the above notations and the assumption that $\psi_\CA(x)$ holds:
\begin{align*}
\alpha_n = \Fmu(\CS^n_0)\text{ for even $n$}\quad\text{and}\quad
\alpha_n = \Fmu(\CR^n_0)\text{ for odd $n$.}
\end{align*}
\end{lemma}

\begin{proof}
The proof is inductive in $n$. First, $\alpha_0=\Fmu(\CS^n_0)$ because of~\eqref{eq:it-alphaO} and the statement for $n=0$ in Lemma~\ref{lem:S-R-limit-steps} (we can take $\theta=\beta_0$ and $\gamma_\ell=\iota_\ell$ for $\ell=0,\ldots,N$ to check that Condition~\eqref{eq:it-alphaO} holds).

Now assume that the above conditions are true for $n{-}1$ for some $n\in\{1,\ldots,N\}$. Again, by the symmetry we assume that $n$ is odd, i.e.~$\alpha_{n-1} = \Fmu(\CS^{n-1}_0)$. By Lemma~\ref{lem:ind-step-a-to-b} we know that $\beta_{n-1} = \Fmu(\CS^{n-1}_\infty)$. Condition~\eqref{eq:it-b-to-a-o} says that $\alpha_n = \CQ_{< n}(\beta_{n-1})=\CQ_{< n}\big(\Fmu(\CS^{n-1}_\infty)\big)$. Now Lemma~\ref{lem:compute-apply-s-r} implies that $\CQ_{< n}\big(\Fmu(\CS^{n-1}_\infty)\big)=\Fmu\big(\CR^n_0\big)$ and the induction step is complete.
\end{proof}

Equation~\eqref{eq:req-for-ab} together with Condition~\eqref{eq:it-x}, imply that $x=\mu_0\{t\in\trees_A\mid q_\init\in\CS^N_0[t]\}$. Since $N>\Omega(q_\init)$ is even, Lemma~\ref{lem:S-R-limit-steps} implies that $\CS^N_0(q_\init)=\lang(\CA,q_\init)$ and therefore, $q_\init\in\CS^N_0[t]$ if and only if $t\in\lang(\CA)$. This guarantees that $x=\mu_0\big(\lang(\CA)\big)$.

We will now prove completeness of the formula: if $x=\mu_0\big(\lang(\CA)\big)$ then $\psi_\CA(x)$ holds. Choose the distributions $(\alpha_n,\beta_n)_{n=0,\ldots,N}$ in~\eqref{eq:it-ex-ab} as in~\eqref{eq:req-for-ab}. We will show that then the rest of the formula holds. Consider any distribution $\theta$. For each $n=0,\ldots,N$ let $\iota_n$ be an~upward\=/closed set witnessing that $\alpha_n\npreceq\theta$ for $n$ odd (resp.~$\theta\npreceq \alpha_n$ for $n$ even); or any upward closed set if the respective inequality holds.

Take any $(\gamma_n)_{n=0,\ldots,N}$ that are upward closed. We need to check that the sub\=/formula starting in~\eqref{eq:it-alphaO} holds. Conditions~\eqref{eq:it-alphaO} ---~\eqref{eq:it-order-ab-e} and~\eqref{eq:it-x} hold by the same lemmata as mentioned in the previous section. To check Conditions~\eqref{eq:it-order-bt-o} and~\eqref{eq:it-order-bt-e} one again invokes Proposition~\ref{pro:compute-limits}: either $\iota_n$ witnesses that $\alpha_n\npreceq\theta$ (resp. $\theta\npreceq \alpha_n$) or, if $\iota_n$ was chosen arbitrarily, then Proposition~\ref{pro:compute-limits} implies that also the respective inequality with $\beta_n$ holds.
\end{proof}

\section{Branching processes}
\label{sec:branching}

For the sake of~simplicity we~define only binary branching processes, the case of~a~fixed higher arity can be~solved analogously. A~\emph{branching process} is~a~tuple $\CP = \langle A, \tau, \alpha_\init   \rangle$ 
where $A$ is~a~finite alphabet; $\fun{\tau}{A}{\Fdis{A^2}}$ a~\emph{branching function} that assigns a~probability distribution over~$A^2$ to~every letter in~$A$;  and $\alpha_\init \in \Fdis{A}$  an~\emph{initial distribution}. 
We~assume that all probabilities occurring in these distributions are rational. 
By~the \emph{size} of~$\CP$ we~understand the size of its binary representation.

A~branching process~$\CP$ can be~seen as~a~generator of~random trees: it~defines a~complete Borel measure~$\mu_\CP$ over the set of~infinite trees in~the following way.
Let $\fun{f}{\dom(f)}{A}$ be~a~complete finite tree of~depth~$d\geq 0$ i.e.~$\dom(f)=\{u\in \{\dL,\dR\}^*\mid |u|\leq d\}=\{\dL,\dR\}^{< d+1}$. Then the measure~$\mu_{\CP}$ of~the basic set~$U_{f}$, see Section~\ref{sec:basic}, is~defined by
\begin{equation}
\mu_\CP(U_{f}) \eqdef \alpha_\init(f(\varepsilon)) \cdot \prod_{u \in \{\dL,\dR\}^{<d}} \tau(f(u))\big(f(u\dL),f(u \dR)\big).
\label{eq:def-of-process-measure}
\end{equation}
Now, $\mu_\CP$~can be~extended in~a~standard way to~a~complete Borel measure on~the set of~all infinite trees~$\trees_A$. Intuitively, a~tree~$t\in\trees_A$ that is~chosen according to~$\mu_\CP$ is~generated in~a~top\=/down fashion: the root label~$t(\varepsilon)$ is~chosen according to~the initial distribution~$\alpha_\init$; and the labels of~the children~$u\dL$ and $u\dR$ of~a~node $u$ are chosen according to~the distribution $\tau(t(u))\in\Fdis{A^2}$ defined %by
for the label of~their parent~$u$.

Observe that the uniform measure~$\mu_0$ over trees~$\trees_A$ equals the measure~$\mu_{\CP_0}$ defined by~the branching process
$\CP_0 = \langle A, \tau_0, \alpha_0 \rangle$,
where $\alpha_0(a) = |A|^{-1}$ and $\tau_0(a)(a_\dL,a_\dR) = |A|^{-2}$ for each $a,a_\dL,a_\dR\in A$.

\begin{theorem}
\label{thm:exists-psi-branching}
Given a~weak alternating automaton~$\CA$ and a~branching process~$\CP$ one can compute a~formula~$\psi_{\CA,\CP}(x)$ that 
%holds in $\R$ for a~unique number $x=\mu_{\CP}\big(\lang(\CA)\big)$. 
represents the number~$\mu_{\CP}\big(\lang(\CA)\big) $. 
Moreover, the formula is~in a~prenex normal form; its size is~exponential in~the size of~$\CA$ and polynomial in~the size of~$\CP$; and the quantifier alternation of~$\psi_{\CA,\CP}$ is~constant.
\end{theorem}

If~one does not care about the complexity, the above result can be~obtained directly, by~interpreting the branching process $\CP$~in an~automaton $\CA$. More precisely, there exists an~algorithm that, given a~weak alternating automaton $\CA$ and a~branching process~$\CP$, computes another weak alternating automaton~$\CA_{\CP}$ such that
\[\mu_{\CP}\big(\lang(\CA)\big)=\mu_{0}\big(\lang(\CA_\CP)\big).\]
Therefore, the decidability part of~Theorem~\ref{thm:exists-psi-branching} follows directly from Theorem~\ref{thm:exists-psi}.
A~construction of~$\CA_\CP$ is given in Subsection~\ref{ssec:reduction}. Another advantage of the construction given there is that it deals explicitly with branching processes of arbitrary branching (possibly non\=/binary). However, it~is~possible to~provide a~direct way of~constructing the formula~$\psi_{\CA,\CP}$ with the size of~the formula polynomial in~the size of~$\CP$, see Subjection~\ref{ssec:direct}.

\subsection{Encoding branching processes in automata}
\label{ssec:reduction}

This section shows how to use the expressive power of weak MSO to simulate branching processes within the uniform measure.

An~\emph{$\ell$\=/branching tree} over an~alphabet $A$ is a~function $\fun{t}{\{\dD_1,\ldots,\dD_\ell\}^*}{A}$, where $\dD_1$,\ldots,$\dD_\ell$ are $\ell$ distinct symbols (we assume that $\dL=\dD_1$ and $\dR=\dD_2$). The set of all such trees is denoted $\trees_A^{(\ell)}$.

Similarly, an~\emph{$\ell$\=/branching process} $\CP=\langle A, \tau, \alpha_\init\rangle$ is defined analogously to a~branching process, except that a~branching function $\fun{\tau}{A}{\Fdis{A^\ell}}$ randomly produces $\ell$\=/tuples of letters. This implies that the measure $\mu_\CP$ is a~Borel measure over the set of $\ell$\=/branching trees $\trees_A^{(\ell)}$.

An~\emph{$\ell$\=/branching} alternating automaton $\CA$ is again analogous to a~standard alternating automaton but the atoms $(d,q')$ in the transition formulae satisfy $d\in\{\dD_1,\ldots,\dD_\ell\}$. If $t$ is an~$\ell$\=/branching tree and $\CA$ is an~$\ell$\=/branching automaton, then the game $\CG(t,p)$ is defined analogously as in Section~\ref{sec:basic}. Thus, the language $\lang(\CA)$ is a~subset of $\trees_A^{(\ell)}$.

According to the above definitions, standard trees, branching processes, and automata, as defined in the main body of this article, are $2$\=/branching.

\begin{proposition}
\label{pro:reductionToUniformMeasure}
Let $\CA$ be a~weak $\ell$\=/branching alternating automaton over an~alphabet $A$ and $\CP$ be~a~$\ell$\=/branching process. Let $A_0$ be any alphabet with at least two symbols. Then, one can construct a~weak $2$\=/branching alternating automaton $\CA_\CP$ over the alphabet $A_0$ such that $\mu_{\CP}\big(\lang(\CA)\big) = \mu_0\big(\lang(\CA_\CP)\big)$, where $\mu_0$ is the uniform measure over $2$\=/branching trees $\trees_{A_0}$.
\end{proposition}

Notice that for $\ell=2$ this reduction is made redundant by the results of Subsection~\ref{ssec:direct}, which allows us to directly compute $\mu_\CP\big(\lang(\CA)\big)$. Moreover, the construction provided there has better complexity: the obtained formula $\psi_{\CA,\CP}$ is only polynomial in the size of $\CP$. However, we provide the present reduction because it shows that the class of languages recognisable by weak alternating automata is robust. In particular, if one does not care about the size of the respective formulae, then Theorem~\ref{thm:exists-psi-branching} can be obtained via the above reduction directly from Theorem~\ref{thm:exists-psi}. Also, this is the only place in the article when we explicitly deal with branching processes of higher branching than $2$.

We~start with an~encoding of~rational numbers.

\begin{lemma}
\label{lem:encoding-rational-in-words}
Let $X$ be a~finite set, $A_0$ any alphabet with at least two symbols, and $\alpha\in\Fdis{X}$ a~probabilistic distribution with rational values. Then there exists a~weak alternating automaton $\CA_\alpha$ over the alphabet $A_0$ with a~set of states $Q_\alpha$ and a~function $\fun{j}{X}{Q_\alpha}$ such that:
\begin{itemize}
\item for $x\neq x'\in X$ the languages $\lang(\CA_\alpha,j(x))$ and $\lang(\CA_\alpha,j(x'))$ are disjoint;
\item the union $\bigcup_{x\in X} \lang(\CA_\alpha,j(x))$ is the set of all trees $\trees_{\{0,1\}}$;
\item for every $x\in X$ the measure $\mu_0\big(\lang(\CA_\alpha,j(x))\big)$ equals $\alpha(x)$.
\end{itemize}
\end{lemma}

\begin{proof}
Without loss of generality we can assume that $A_0=\{0,\ldots, | A_0| -1\}$.
Assume that $X=\{x_1,\ldots,x_K\}$. Fix rational numbers $r_k\eqdef\sum_{k'\leq k}\alpha(x_{k'})$ for $k=0,\ldots, K$. We know that $r_0=0$ and $r_K=1$. For each $k=0,\ldots,K$ let $e_k$ be the $M$\=/ary expansion of $r_k$, i.e.~$e_k \in A_0^{\omega}$ is~a~word such that~$r_k = 0.e_k$. Since each of the numbers $r_k$ is rational, the words $e_k$ are ultimately periodic, i.e.~of the form $u\cdot v\cdot v \cdot v \cdots$

Each tree $t\in\trees_{A_0}$ induces a~real number $r(t)\in[0,1]$ that
is obtained by reading the left\=/most branch of $t$ and treating it 
as an~$|A_0|$\=/ary expansion of $r(t)$.

Let $e,e'\in A_0^\omega$ be two expansions of rational numbers with
$0.e< 0.e'$. It is now standard 
%\MS{add citation?} 
to construct a~weak deterministic automaton $\CA_{e,e'}$ with an~initial state $q_{e,e'}$ that accepts a~tree $t\in\trees_{A_0}$ if and only if $0.e\leq r(t)<0.e'$.

Now, to obtain the automaton $\CA_\alpha$ it is enough to take the disjoint union of the automata $\CA_{e_{k-1},e_k}$ for $k=1,\ldots,K$ and define $j(x_k)=q_{e_{k-1},e_k}$.
\end{proof}

We now move to the proof of Proposition~\ref{pro:reductionToUniformMeasure}. Take a~weak $\ell$\=/branching alternating automaton over an alphabet $A$ and an~$\ell$\=/branching process $\CP$ over the same alphabet. For the sake of simplicity assume that the initial distribution $\alpha_\init$ of $\CP$ is concentrated in a~single letter $a_\init \in A$.

The above construction will be used to simulate the random choice represented by the distributions $\tau(a)\in \Fdis{A^\ell}$. The automaton $\CA_\CP$ is defined as~a disjoint union of the automata $\CA_{\tau(a)}$ for each $a\in A$ together with a~modified copy of $\CA$. This modified copy of $\CA$ has states of the following two forms:
\begin{itemize}
\item pairs $(q,a)$ where $q$ is a state of $\CA$ and $a\in A$;
\item triples $(d,q,a)$ where $d\in\{\dD_1,\ldots,\dD_\ell\}$, $q$ is a~state of $\CA$, and $a\in A$.
\end{itemize}

Given a~transition $\delta(q,a)$ of the automaton $\CA$ and a~vector $\vec{a}\in A^\ell$ let $\bar{\delta}(q,a,\vec{a})$ be defined as the same formula as $\delta(q,a)$, except that each atom $(d,q)$ is replaced by $\big(\dR,(d,q,\vec{a}(d))\big)$ --- a~transition to the right in a~tree to the state $(d,q,\vec{a}(d))$ of $\CA_\CP$. Now, the automaton $\CA_\CP$, together with all the transitions of $\CA_{\tau(a)}$ for $a\in A$ has the following transitions for $b\in A_0$:
\begin{align*}
\delta\big((q,a),b\big)&\eqdef \bigvee_{\vec{a}\in A^\ell}\ \big(\dL, j(\vec{a})\big)\land \bar{\delta}(q,a,\vec{a})\\
&\qquad\text{where $j(\vec{a})$ is the respective state of the automaton $\CA_{\tau(a)}$}\\
&\qquad\text{such that $\mu_0\big(\lang(\CA_{\tau(a);}, j(\vec{a})\big)=\tau(a)(\vec{a})$}\\
\delta\big((\dD_1,q,a),b\big)&\eqdef \big(\dL, (q,a)\big)\\
\delta\big((\dD_{k+1},q,a),b\big)&\eqdef \big(\dR, (\dD_k,q,a)\big)\qquad\text{for $k=1,\ldots,\ell{-}1$.}\\
\end{align*}
The priority mapping of $\CA_\CP$ is taken from $\CA_{\tau(a)}$ and $\CA$ respectively, i.e.~$\Omega(q,a)=\Omega(d,q,a)=\Omega(q)+2$ --- we need this shift because the initial states of $\CA_{\tau(a)}$ have priority $2$. Let the initial state of $\CA_\CP$ be $(q_\init,a_\init)$.

The automaton $\CA_\CP$ is designed in such a~way to treat each tree $t\in\trees_{A_0}$ as an~encoded version of a~tree $t\in\trees_{A}$. To formally prove this fact, we first need to define that encoding. For this purpose, we define a~family of functions $T_a$ from $\trees_{A_0}$ into $\trees_{A}^{(\ell)}$ indexed by letters $a\in A$. Consider $a\in A$ and a~tree $t\in\trees_{A_0}$. Let $t'\eqdef t\restr_\dL$ be the left subtree of $t$. Similarly, for $k=1,\ldots,\ell$ let $t_k\eqdef t\restr_{\dR^k\dL}$. Let $\vec{a}\in A^\ell$ be the unique vector of letters such that $t'\in\lang\big(\CA_{\tau(a)},j(\vec{a})\big)$. Notice that since $t$ was chosen randomly, the probability distribution of the vectors $\vec{a}$ defined here is exactly $\tau(a)$. Then, let the resulting tree $T_a(t)$ have the root labelled $a$ and for $k=1,\ldots,\ell$ let the $\dD_k$\=/th subtree of $T_a(t)$ equal $T_{\vec{a}(k)}(t_k)$. See Figure~\ref{fig:the-operator-T} for a~depiction of that definition.

\begin{figure}
\centering
\begin{tikzpicture}[scale=0.6]

\newcommand{\subtr}[3]{
\coordinate (#1) at (#2) {};

\draw (#1) -- ++(-0.5,-1);
\draw (#1) -- ++(+0.5,-1);
\node[anchor=center] at ($(#1)+(0,-0.8)$) {#3};
}

\newcommand{\subTR}[3]{
\coordinate (#1) at (#2) {};

\draw (#1) -- ++(-1.2,-2.4);
\draw (#1) -- ++(+1.2,-2.4);
\node[anchor=center] at ($(#1)+(0,-2.0)$) {#3};
}

\node (t) at (0,+1.5) {$t$};
\node (T) at (10,+1.5) {$T_a(t)$};
\draw (t) edge[|-Latex, bend left=10] (T);

\tikzstyle{treeN} = [draw,circle]
\tikzstyle{treeE} = [draw]
\tikzstyle{treeD} = [line cap=round,thick,dash pattern=on \pgflinewidth off 12pt]

\node[treeN] (r) at (0,0) {$\_$};
\subtr{tl}{$(r)+(-1.5,-2)$}{$t'$}
\draw[treeE] (r) -- (tl);

\node[treeN] (t1) at ($(r)+(+1,-2)$) {$\_$};
\draw[treeE] (r) -- (t1);
\subtr{tr1}{$(t1)+(-1,-2)$}{$t_1$}
\draw[treeE] (t1) -- (tr1);
\draw[treeE] (t1) -- ++(0.6,-1.2);
\coordinate (dds) at ($(t1)+(+1.2,-2.4)$);

\draw[treeD] ($(dds)+(-0.3,+0.6)$) -- ($(dds)+(+0.3,-0.6)$);

\coordinate (dds) at ($(t1)+(+1.2,-2.4)+(-1,-2)$);

\draw[treeD] ($(dds)+(-0.3,+0.6)$) -- ($(dds)+(+0.3,-0.6)$);

\node[treeN] (t2) at ($(t1)+(+2,-4)$) {$\_$};

\subtr{tr2}{$(t2)+(-1,-2)$}{$t_\ell$}
\subtr{te}{$(t2)+(+1,-2)$}{$\_$}

\draw[treeE] (t2) -- (tr2);
\draw[treeE] (t2) -- (te);

\node[treeN] (r) at (10,0) {$a$};
\subTR{tr1}{$(r)+(-2.5,-2)$}{$T_{\vec{a}_1}(t_1)$};
\subTR{tr2}{$(r)+(+2.5,-2)$}{$T_{\vec{a}_\ell}(t_\ell)$};
\draw[treeE] (r) -- (tr1);
\draw[treeE] (r) -- (tr2);

\draw[treeD] ($(r)+(-1,-2.5)$) -- ($(r)+(+1,-2.5)$);
\end{tikzpicture}
\caption{An~illustration of an~operation $T_a$ for $a\in A$. Nodes and the subtree marked with $\_$ are irrelevant in this construction. The subtree $t'$ is used to determine which vector $\vec{a}\in A^\ell$ to use --- it simulates the random choice of that vector using $\tau(a)$. Then the subtrees $t_k$ for $k=1,\ldots,\ell$ are recursively decoded by $T_{\vec{a}_k}$ according to the chosen letters of $\vec{a}$.}
\label{fig:the-operator-T}
\end{figure}

\begin{claim}
\label{cl:T-a-red}
Given a~tree $t\in\trees_{A_0}$ the automaton $\CA_\CP$ accepts $t$ from a~state $(q,a)$ if and only if $\CA$ accepts the tree $T_a(t)$ from $q$. In other words,
\[\lang\big(\CA_\CP,(q,a)\big)=T_a^{-1}\big(\lang(\CA,q)\big).\]
\end{claim}

\begin{proof}
First observe that Lemma~\ref{lem:encoding-rational-in-words} implies that whenever $\CA_\CP$ takes a~transition of the form $\delta\big((q,a),b\big)$ then there is exactly one candidate of $\vec{a}\in A^\ell$ such that the left subtree under the current node can be accepted from the state $j(\vec{a})$. Therefore, player $\eve$ in the game $\CG\big(t,(q,a)\big)$ is always forced to choose that disjunct there. If the proper disjunct is chosen, then the choice of the atom $\big(\dL,j(\vec{a})\big)$ is losing for \adam because the respective subtree $t'$ belongs to $\lang\big(\CA_{\tau(a)},j(\vec{a})\big)$. Thus, we can assume that \adam never chooses this atom.

Under the two above assumptions, the game $\CG\big(t,(q,a)\big)$ given by the automaton $\CA_\CP$ becomes equivalent to the game $\CG(T_a(t), q)$ given by the automaton $\CA$.
\end{proof}

The next lemma states that the mapping $T_{a_\init}$ for the initial symbol $a_\init\in A$ allows to move between the measures $\mu_0$ and $\mu_\CP$. Recall that we have assumed that $\alpha_\init(a_\init)=1$.

\begin{lemma}
\label{lem:T-a-pres}
The mapping $T_{a_\init}$ preserves the measure: for every measurable subset $L\subseteq \trees_A^{(\ell)}$ and its pre\=/image $L'\eqdef T_{a_\init}^{-1}(L)$ we have $\mu_0(L')=\mu_\CP(L)$.
\end{lemma}

\begin{proof}
It is enough to check this on a~basic set $L$ as in~\eqref{eq:def-of-process-measure}. But in that case it follows from Lemma~\ref{lem:encoding-rational-in-words} and the fact that the subtrees $t'$ used to choose the respective vectors $\vec{a}$ have pairwise\=/incomparable roots.
\end{proof}

By applying Claim~\ref{cl:T-a-red} and Lemma~\ref{lem:T-a-pres} we obtain that \[\mu_0\big(\lang(\CA_\CP)\big)=\mu_0\big(\lang\big(\CA_\CP,(q_\init,a_\init)\big)\big)=\mu_0\big(T_{a_\init}^{-1}\big(\lang(\CA,q_\init)\big)\big)=\mu_\CP\big(\lang(\CA,q_\init)\big)=\mu_\CP\big(\lang(\CA)\big).\]
This concludes the proof of Proposition~\ref{pro:reductionToUniformMeasure}.

\subsection{Branching processes - direct construction}
\label{ssec:direct}

\newcommand{\muBP}{\mu_{\CP}}
\newcommand{\muBPvec}{\vec{\mu}_{\CP}}
\newcommand{\alphavec}{\alpha_\CP}
\newcommand{\betavec}{\beta_\CP}

In this section we want to show how to extend our main result, of computing the uniform measure of a~weak\=/MSO recognisable language, to
measures generated by arbitrary branching processes.

The core of the proof will stay the same as in the main part of~the article, we will define two types of operators $\CF$, $\CQ$,
and explain, how the measure can be computed using those operators.

Let us fix a~regular language of trees $L$ and a~weak alternating automaton $\CA$ such that $\lang(\CA) = L$.

Let us fix a~branching process $\CP=\langle A,\tau,\alpha_\init\rangle$. We want to distinguish between the alphabet $A$ treated as the set of labels of trees, from $A$ treated as vertices of the branching process $\CP$. Thus, we put $V=A$ and use the symbol $v\in V$ to denote letters generated by $\CP$. This means that $\fun{\tau}{V}{\Fdis{V^2}}$ and $\alpha_\init\in\Fdis{V}$.

By $\muBP(v)$, where $v \in V$,  we understand the measure induced by the process $\CP$ with the initial distribution $\alpha'_\init$ concentrated in $v$, i.e.~$\alpha'_\init(v) = 1$ and $\alpha'_\init(v')=0$ for $v'\neq v$.

By a~simple calculation, we have that
\begin{equation}
\label{eq:measure-bp-to-dirac-bp}
\mu_{\CP}(L) = \sum_{v \in V} \alpha_\init(v)\cdot \mu_{\CP}(v)(L).
\end{equation}

Thus, we only need to determine the values of $\muBPvec(v)(L)$ for $v\in V$.

The measure defined in a subtree, unlike in Remark~\ref{rem:measure}, is not always uniform and may non\=/trivially depend on~the label of the root of the subtree.
This implies that the distributions~$\beta$ used in the whole procedure may depend on the initial vertex and, thus, this information has to be included.
It turns out that simply lifting distributions to tuples indexed by the origin point in the branching process is enough.
We lift distributions to tuples of distributions by defining $\fun{{\betavec}}{V}{\Fdis \powerset(Q)}$. In other words, the basic space that we work is, instead of $\Fdis\powerset(Q)$ is now $\big(V\to \Fdis\powerset(Q)\big)$. Let the order ${\preceq}$ be defined on $V\to \Fdis\powerset(Q)$ coordinate-wise: $\alphavec\preceq \betavec$ if for every $v\in V$ we have $\alphavec(v) \preceq \betavec(v)$.

Now, our definitions of previously used operations have to be adjusted accordingly.
By slight abuse of notation, we will simply overload the definitions. This will not produce confusion, since we~will not use the old definitions in this part.

Take a~$Q$\=/indexed family $\CL$. Define the distribution ${\muBPvec}\in V \to \Fdis P(Q)$.
\begin{equation}
\muBPvec\big(\CL\big){(v)}(P)\eqdef\muBP{(v)}\big\{t\in\trees_A\mid \CL[t]=P\big\}
\end{equation}
Notice that the set of trees with root labelled $v$ is of full ${\mu}_{\CP}(v)$ measure. Thus
\begin{equation}
\label{eq:root-label-fixed}
\mu_{\CP}(v)\big\{t\in\trees_A\mid \CL[t]=P\big\} = \mu_{\CP}(v)\big\{t\in\trees_A\mid \CL[t]=P \land t(\varepsilon) = v\big\}
\end{equation}
Also, the measure $\mu_{\CP}$ satisfies the following independence property similar to Remark~\ref{rem:measure}.

\begin{remark}
\label{rem:CP-local-ind}
Let $L_\dL, L_\dR\subseteq\trees_A$ be two Borel sets and $v\in V$. Then
\[\mu_{\CP}(v)\big\{t\mid t\restr_\dL {\in} L_\dL\land t(\varepsilon){=}v\land t\restr_\dR{\in} L_\dR\big\}=\sum_{v_\dL, v_\dR\in V^2}\tau(v)(v_\dL,v_\dR)\cdot\mu_{\CP}(v_\dL)\big(L_\dL\big)\cdot \mu_{\CP}(v_\dR)\big(L_\dR\big).\]
\end{remark}

%If the process $\CP$ is clear from the context, we may drop it and simply write $\vec{\mu}_{v}$ and $\mu_{v}$.
As before, the sets in consideration are measurable thanks to Proposition~\ref{pro:measurability}.

\begin{lemma}
\label{lem:order-of-families-bp}
Fix $v \in V$.	If for every $q\in Q$ we have $\CL(q)\subseteq \CL'(q)$ then $\vec{\mu}_{\CP}(\CL)\preceq \vec{\mu}_{\CP}(\CL')$ in $V\to\Fdis{\powerset(Q)}$.
\end{lemma}
\noindent
The proof is the same as the proof of Lemma~\ref{lem:order-of-families}, as it depends on general properties of~measures.

Now, we examine the sequences of distributions $\vec{\mu}_{\CP}\big(\CS^n_i\big)$, $\vec{\mu}_{\CP}\big(\CR^n_i\big)$, $\vec{\mu}_{\CP}\big(\CS^n_\infty\big)$, and $\vec{\mu}_{\CP}\big(\CR^n_\infty\big)$ arising from the $Q$\=/families introduced before.
Our aim again is to bind them by some equations computable within $V\to\Fdis{\powerset(Q)}$.
As an analogue to the operation $\CF$, we introduce the function  $\fun{\CF_{\CP}}{(V \to \Fdis{\powerset(Q)})}{(V \to \Fdis{\powerset(Q)})}$ defined for $v \in V$, $\betavec \in V \to \Fdis{\powerset(Q)}$, and $P\in \powerset(Q)$ by
\begin{equation}
\label{eq:def-F-bp}
\CF_{\CP}(\betavec)(v)(P)\eqdef \sum_{(P_\dL,v,P_\dR)\in\Delta^{-1}(P)}\ %
                  \sum_{(v_\dL, v_\dR) \in V^2} \tau(v)(v_\dL, v_\dR)\big(\betavec(v_\dL)(P_\dL)%
                  \cdot \betavec(v_\dR)(P_\dR)\big)
\end{equation}

As before, the formula  guarantees that $\CF_{\CP}(\betavec)(v)$ is indeed a~probabilistic distribution in $\Fdis\powerset(Q)$.
The operator $ \CF$  will allow us to transfer the inductive
definitions of the $Q$\=/families $\CS^n_{i+1}$ and
$\CR^n_{i+1}$ given by Lemma~\ref{lem:ind-for-S-R-Delta}, to the level
of probability distributions.

From now on, we omit the index $\CP$ in $\CF_{\CP}$.

\begin{lemma}
	\label{lem:step-apply-F-bp}
		For each $n\in\N$ and $i\in\N$ we have
	\[\muBPvec\big(\CS^n_{i+1}\big)=\CF\Big(\muBP\big(\CS^n_i\big)\Big)\text{ and }\muBPvec\big(\CR^n_{i+1}\big)=\CF\Big(\muBPvec\big(\CR^n_i\big)\Big).\]
\end{lemma}

\begin{proof}
	Take $P\in\powerset(Q)$ and $v \in V$ observe that
%	{\tiny
	\begin{align*}
		\CF\Big(\muBPvec\big(\CS^n_i\big)\Big)(v)(P) &\eqext{1} \sum_{(P_\dL,v,P_\dR)\in\Delta^{-1}(P)}\ %
		\sum_{(v_\dL, v_\dR) \in V^2} \tau(v)(v_\dL, v_\dR)\cdot\\
		&\qquad\qquad\qquad
		\Big(\muBPvec\big(\CS^n_i\big)({v_\dL})(P_\dL)%
		\cdot \muBPvec\big(\CS^n_i\big)({v_\dR})(P_\dR)\Big)\\
		%%%%%%%%%%%%%%%%%%%%%%%%%%%%%%%%%%%%%%%%%%%%%%%%%%%%%%%%%%%%%%%%%%%%%%%%%%%%%%
		& \eqext{2} \sum_{(P_\dL,v,P_\dR)\in\Delta^{-1}(P)}\ %
		\sum_{(v_\dL, v_\dR) \in V^2} \tau(v)(v_\dL, v_\dR)\cdot\\
		&\qquad\qquad\qquad
		\Big(\muBP({v_\dL})\big\{t_\dL\mid \CS^n_i[t_\dL]{=}P_\dL\big\}%
		\cdot \muBP({v_\dR})\big\{t_\dR\mid \CS^n_i[t_\dR]{=}P_\dR\big\}\Big)\\
		%%%%%%%%%%%%%%%%%%%%%%%%%%%%%%%%%%%%%%%%%%%%%%%%%%%%%%%%%%%%%%%%%%%%%%%%%%%%%%
		&\eqext{3} \sum_{(P_\dL,v,P_\dR)\in\Delta^{-1}(P)}\ \mu_\CP(v)\big\{t\mid \CS^n_i[t\restr_\dL]{=}P_\dL\land t(\varepsilon){=}v\land \CS^n_i[t\restr_\dR]{=}P_\dR\big\}\\
		%%%%%%%%%%%%%%%%%%%%%%%%%%%%%%%%%%%%%%%%%%%%%%%%%%%%%%%%%%%%%%%%%%%%%%%%%%%%%%
		&\eqext{4}\muBP(v)\left(\bigcup_{(P_\dL,v,P_\dR)\in\Delta^{-1}(P)}\big\{t\mid \CS^n_i[t\restr_\dL]{=}P_\dL\land t(\varepsilon){=}v\land \CS^n_i[t\restr_\dR]{=}P_\dR\big\}\right)\\
		%%%%%%%%%%%%%%%%%%%%%%%%%%%%%%%%%%%%%%%%%%%%%%%%%%%%%%%%%%%%%%%%%%%%%%%%%%%%%%
		&\eqext{5} \muBP(v)\Big\{t\in\trees_A\mid \Delta\big(\CS^n_i[t\restr_\dL], t(\varepsilon), \CS^n_i[t\restr_\dR]\big){=}P\land t(\varepsilon)=v\Big\}\\
		%%%%%%%%%%%%%%%%%%%%%%%%%%%%%%%%%%%%%%%%%%%%%%%%%%%%%%%%%%%%%%%%%%%%%%%%%%%%%%
		&\eqext{6} \muBP(v)\Big\{t\in\trees_A\mid \CS^n_{i+1}[t]{=}P\land t(\varepsilon)=v\Big\}\\
		&\eqext{7} \muBP(v)\Big\{t\in\trees_A\mid \CS^n_{i+1}[t]{=}P\Big\}\eqext{8}\muBPvec\big(\CS^n_{i+1}\big)(v)(P),
	\end{align*}
%	}%
where:~(1) is just the definition of
$\CF\Big({\muBPvec}\big(\CS^n_i\big)\Big)$; (2) follows from the definition
of $\muBPvec\big(  \CS^n_i \big)$; (3) follows from 
the definition of $\muBP$ and Remark~\ref{rem:CP-local-ind};
(4)~follows from the fact that the measured sets are pairwise
disjoint; (5)~follows 
%from the fact that $\Delta$ is a~function; 
simply from the definition of $\Delta$;
(6)~follows from Lemma~\ref{lem:ind-for-S-R-Delta}; (7) follows from~\eqref{eq:root-label-fixed}; and (7) is just the definition of $\muBPvec \big(\CS^n_{i+1}\big)$.
	
The proof for $\CR^n_{i+1}$ is entirely analogous (we use the $\CR^n_i$ variant of Lemma~\ref{lem:ind-for-S-R-Delta} instead).
\end{proof}

Now, recall that $Q_{\geq n}$ and $Q_{<n}$ are sets of states of
respective priorities. Let 
the functions $\fun{\CQ_{\geq n},\allowbreak\CQ_{<n}}{(V \to\Fdis\powerset(Q))}{(V \to\Fdis\powerset(Q))}$ be defined by
\begin{align*}
	\CQ_{\geq n}({\betavec})(v)(P)&\eqdef \sum_{P'\colon  P'\cup Q_{\geq n}=P} {\betavec}(v)(P'),\\
	\CQ_{< n}({\betavec})(v)(P)&\eqdef \sum_{P'\colon P'\cap Q_{< n}=P}{\betavec}(v)(P').
\end{align*}

\noindent
Again, the formulas guarantee that $\CQ_{\geq n}({\betavec})(v)$ and $\CQ_{< n}({\betavec})(v)$ are both probabilistic distributions in $\Fdis\powerset(Q)$.
The following lemma shows the relation between these functions and the limit distributions ${\muBPvec}\big(\CS^{n-1}_{\infty}\big)$ and ${\muBPvec}\big(\CR^{n-1}_{\infty}\big)$.

\begin{lemma}
	\label{lem:compute-apply-s-r-bp}
	For each $n\in\N$ we have
	\begin{align*}
		\CQ_{< n}\Big({\muBPvec}\big(\CS^{n-1}_\infty\big)\Big)&={\muBPvec}\big(\CR^n_0\big)&\text{if $n$ is odd,}\\
		\CQ_{\geq n}\Big({\muBPvec}\big(\CR^{n-1}_\infty\big)\Big)&={\muBPvec}\big(\CS^n_0\big)&\text{if $n$ is even.}
	\end{align*}
\end{lemma}

\begin{proof}
Note that the proof of Lemma~\ref{lem:compute-apply-s-r} does not depend on the underlying measure, and therefore it carries over.
\end{proof}

Again, the two above lemmata show that the operators $\CF$, $\CQ_{< n}$, and $\CQ_{\geq n}$ are enough to~perform the respective computations in $V\to\Fdis\powerset(Q)$ as they do on Figure~\ref{fig:construction}.

In Section~\ref{sec:limit} we prove the connection between the limit distributions ${\muBPvec}\big(\CS^{n}_\infty\big)$, ${\muBPvec}\big(\CR^{n}_\infty\big)$ (unary versions) and fixed points of the operator $\CF$, see Lemma~\ref{granica}.
The same line of proof works in the case of tuples, if we apply the reasoning point\=/wise, i.e.~we work now in $(\R^{\powerset(Q)})^V=\R^{\powerset(Q)\times V}$.
The only missing ingredient is the monotonicity of the new operator~$\CF_\CP$.

\begin{lemma}
\label{lem:props-of-F}
The operator $\fun{\CF}{(V \to \Fdis{\powerset(Q)})}{(V \to \Fdis{\powerset(Q)})}$ is point\=/wise monotone in the order~${\preceq}$ and continuous in $\R^{\powerset(Q)\times V}$.
\end{lemma}

\begin{proof}
Continuity is again trivial. The fact that $\CF$ is monotone follows from the monotonicity of $\Delta$ and the point\=/wise definition of the order as follows. Recall the definition of $\CF_\CP$, cf.~\eqref{eq:def-F-bp}:
\[
\CF(\betavec)(v)(P)= \sum_{(P_\dL,v,P_\dR)\in\Delta^{-1}(P)}\ %
\sum_{(v_\dL, v_\dR) \in V^2} \tau(v)(v_\dL, v_\dR)\big(\betavec(v_\dL)(P_\dL)%
\cdot \betavec(v_\dR)(P_\dR)\big)
\]
We need to prove that for a~fixed $v\in V$ function $\CF$ is monotone w.r.t. the order ${\preceq}$. Thus, for every $\alphavec\preceq \betavec\in V\to\Fdis{\powerset(Q)}$ and an~upward\=/closed family $U\subseteq \powerset(Q)$ we should have $\sum_{P\in U} \CF_\CP(\alphavec)(v)(P)\leq \sum_{P\in U} \CF_\CP(\betavec)(v)(P)$. After splitting the sum over separate letters $v$, $v_{\dL}$, $v_{\dR}\in V$, it is enough to show that for 
$O_v\eqdef \{(P_\dL,P_\dR)\mid\Delta(P_\dL,v,P_\dR)\in U\}$ we have
\[\sum_{(P_\dL, P_\dR)\in O_v}\alphavec(v_{\dL})(P_\dL)\cdot \alphavec(v_{\dR})(P_\dR)\leq\sum_{(P_\dL, P_\dR)\in O_v}\betavec(v_{\dL})(P_\dL)\cdot \betavec(v_{\dR})(P_\dR).\]
		
The set $O_v$ is again upward\=/closed on both coordinates, as in the proof of Lemma~\ref{lem:monotone-F}. We use the notation used there to denote the sections of that set.
Thus, using the assumption that $\alphavec\preceq \betavec$ twice (once for $v_\dL$ and once for $v_\dR$), we obtain

\begingroup
\allowdisplaybreaks
	\begin{align*}
		\sum_{(P_\dL, P_\dR)\in O_v}\alphavec(v_{\dL})(P_\dL)\cdot \alphavec(v_{\dR})(P_\dR)&= \sum_{P_\dL\in\powerset(Q)} \alphavec(v_{\dL})(P_\dL)\cdot \left(\sum_{P_\dR \in P_\dL^{-1}\cdot O_v} \alphavec(v_{\dR})(P_\dR)\right)\\
		&\leq \sum_{P_\dL\in\powerset(Q)} \alphavec(v_{\dL})(P_\dL)\cdot \left(\sum_{P_\dR \in P_\dL^{-1}\cdot O_v} \betavec(v_{\dR})(P_\dR)\right)\\
		&=\sum_{(P_\dL, P_\dR)\in O_v}\alphavec(v_{\dL})(P_\dL)\cdot \betavec(v_{\dR})(P_\dR)\\
		&=\sum_{(P_\dL, P_\dR)\in O_v}\betavec(v_{\dR})(P_\dR)\cdot \alphavec(v_{\dL})(P_\dL)\\
		&=\sum_{P_\dR\in\powerset(Q)} \betavec(v_{\dR})(P_\dR)\cdot \left(\sum_{P_\dL \in O_v\cdot P_\dR^{-1}} \alphavec(v_{\dL})(P_\dL)\right)\\
		&\leq\sum_{P_\dR\in\powerset(Q)} \betavec(v_{\dR})(P_\dR)\cdot \left(\sum_{P_\dL \in O_v\cdot P_\dR^{-1}} \betavec(v_{\dL})(P_\dL)\right)\\
		&=\sum_{(P_\dL, P_\dR)\in O_v}\betavec(v_{\dR})(P_\dR)\cdot \betavec(v_{\dL})(P_\dL)\\
		&=\sum_{(P_\dL, P_\dR)\in O_v}\betavec(v_{\dL})(P_\dL)\cdot\betavec(v_{\dR})(P_\dR).
	\end{align*}%
\endgroup
\end{proof}

Since $\CF$ is continuous and monotone, ${\muBPvec}\big(\CS^{n}_\infty\big)$ 
and ${\muBPvec}\big(\CR^{n}_\infty\big)$ are the greatest\=/ and least\=/fixed points of the appropriate operations.
This observation allows us to compute the values $\mu_{\CP}(v)(L)$ and given Equation~\eqref{eq:measure-bp-to-dirac-bp} 
we obtain the measure $\mu_{\CP}(L)$.

\section{Representing algebraic numbers}
\label{sec:algebra}

\begin{comment}
In~this short subsection we~argue how to~use the formulae~$\psi_\CA$ and $\psi_{\CA,\CP}$ from the above sections to~estimate and even represent the measures of~the language~$\lang(\CA)$.
\end{comment}
We now use  the  formulae~$\psi_\CA$ and $\psi_{\CA,\CP}$ constructed above to find the measure of the language~$\lang(\CA)$.
%Firstly, Theorem~\ref{thm:main-std} is based on the following result.
We~use the celebrated result of~Tarski~\cite{tarski_decision} and its two algorithmic improvements.

%\begin{theorem}[\!\! see~\cite{basu_algebraic_computation}, and~\cite{ben_tarski_complexity}]
\begin{theorem}[\!\!\cite{basu_algebraic_computation,ben_tarski_complexity}]
\label{thm:ben-tarski-MC}
Given a~formula $\psi$ of first\=/order logic over $\R$, one can decide if~$\psi$ holds in deterministic exponential space. Moreover, if~$\psi$ is~in~a~prenex normal form and the alternation of quantifiers $\forall$ and $\exists$ in~$\psi$ is~bounded then the algorithm works in~single exponential time in~the size of~$\psi$.
\end{theorem}

%Based on that we can prove our main Theorem~\ref{thm:main-branching}.

%\begin{proof}[Proof of Theorem~\ref{thm:main-std}]
\begin{proof}[Proof of Theorem~\ref{thm:main-branching}]
Input a~weak alternating automaton $\CA$, a~branching process $\CP$, and a~rational number~$q$. Consider the formula $\psi\equiv \exists x.\ \psi_{\CA,\CP}(x)\land q\bowtie x$, where $\bowtie$ is~one of~${<}$, ${=}$, or ${>}$. Notice that $\psi$ is in prenex normal form; its size is exponential in~the size of~$\CA$ and polynomial in the size of $\CP$; and its quantifier alternation is~constant. Apply the algorithm from Theorem~\ref{thm:ben-tarski-MC} to~check whether $\psi$ is~true in~$\R$.
\end{proof}

%The same method works for Theorem~\ref{thm:main-bra}.

\begin{comment}
However, one may also ask for an~algorithm computing a~representation of~the measure~$\mu_\CP\big(\lang(\CA)\big)$. The celebrated result of~Tarski~\cite{tarski_decision} provides a~\emph{quantifier elimination} procedure: given a~formula $\psi(x_1,\ldots,x_n)$ one can construct an~equivalent quantifier\=/free formula $\widehat{\psi}(x_1,\ldots,x_n)$. Moreover, such a~quantifier free formula can be represented by a~\emph{semialgebraic} set, see~\cite[Chapter~2]{bochnak_real_algebraic}.
\end{comment}

%\noindent
We can also compute  a~representation of~the measure~$\mu_\CP\big(\lang(\CA)\big)$. The 
\emph{quantifier elimination} procedure due to Tarski~\cite{tarski_decision} transforms 
%celebrated result of~ provides a~: given 
a~formula $\psi(x_1,\ldots,x_n)$ into  an~equivalent quantifier\=/free formula $\widehat{\psi} (x_1,\ldots,x_n)$, which moreover
%a~quantifier free formula    $\widehat{\psi} $
can be represented by a~\emph{semialgebraic} set, see  \cite[Chapter~2]{bochnak_real_algebraic}.

\begin{theorem}[\!\!\cite{collins_algebraic_decomposition}]
Given a~formula $\psi(x_1,\ldots,x_n)$ of first\=/order logic over $\R$, one can construct a~representation of the set of tuples $(x_1,\ldots,x_n)$ satisfying $\psi$, as a~semialgebraic set. Moreover, this algorithm works in time doubly\=/exponential in the size of $\psi$.
\end{theorem}

Theorems~\ref{thm:exists-psi} and~\ref{thm:exists-psi-branching} together with the above results imply the following claim. 

\begin{corollary}
Given  a  weak alternating automaton $\CA$ of~size~$n$, one can compute a~representation of~the value~$\mu_0\big(\lang(\CA)\big)$ as~a~singleton semialgebraic set in~time triply exponential in~$n$. Moreover, given a~branching process of size~$m$, one can compute a~representation of~the value~$\mu_\CP\big(\lang(\CA)\big)$ as a~singleton semialgebraic set in~time triply exponential in~$n$ and doubly exponential in~$m$.
\end{corollary}

\begin{comment}
Assume a~given weak alternating automaton $\CA$ of~size~$n$. One can compute a~representation of~the value~$\mu_0\big(\lang(\CA)\big)$ as~a~singleton semialgebraic set in~time triply exponential in~$n$. Moreover, given a~branching process of size~$m$, one can compute a~representation of~the value~$\mu_\CP\big(\lang(\CA)\big)$ as a~singleton semialgebraic set in~time triply exponential in~$n$ and doubly exponential in~$m$.
\end{comment}

\section{Conclusions}
\label{sec:conclusions}

We have~shown  how to~compute the probability measure
%$\mu_0(L)$ 
of~a~tree language~$L$ recognised by~a~weak alternating automaton.
The crucial trait is \emph{continuity} of~certain approximations of
the measure of $L$
%~$\mu_0(L)$
in~a~properly chosen order~${\preceq}$, see Lemma~\ref{granica}. This 
continuity relies on~K\"onig's lemma, cf.~Lemma~\ref{lem:limit-steps-S-R}.
In terms of $\mu$\=/calculus, it stems from both the
absence of 
alternation between least and greatest fixed points  in formulae and
the boundedness of branching  in models
(for a study of continuity in $\mu $\=/calculus see~\cite{fontaine_continuous}).

Whether our  techniques can be extended beyond  weak automata---hopefully 
 to~all tree automata or, equivalently, full MSO logic, or full $\mu$\=/calculus---remains open.  The question is of interest as, e.g.~translation of the logic CTL* into 
 $ \mu$-calculus requires 
%a positive number of alternations 
at least one alternation
between  least and greatest fixed points
(cf.~\cite{demri-temporal-book}, Exercise~10.13).  On the other hand,
fixed point  formulas over binary trees are not continuous in~general, and may require~$\omega_1$ iterations to~reach stabilisation, already on the second level of the fixed\=/point hierarchy.

This problem has been already successfully tackled 
%in~the context of~measurability of~these languages ---
%in the proof 
in the context  of~measurability of regular tree languages---Mio~\cite{mio_game_semantics} uses Martin's axiom to~control the
behaviour of~measure when taking limits of~sequences of~length
$\omega_1$. Such behaviour cannot be~directly simulated in~$\Fdis X$,
because each 
%strictly ascending 
well\=/founded chain of~distributions has a~countable length. However,
this need not be~an~absolute obstacle as~it~might be~the case that the values of~the measure of the iterations stabilise before the actual fixed point is~reached, possibly in~$\omega$ steps.

\clearpage

\bibliography{mskrzypczak}

\end{document}